\documentclass[a4paper,11pt]{article}

\usepackage[utf8]{inputenc}
\usepackage[T1]{fontenc}
\usepackage{lmodern}
\usepackage{amsmath, amssymb, amsfonts, amsthm}
\usepackage{mathtools}
\usepackage{graphicx}
\usepackage{geometry}
\usepackage{xcolor}
\usepackage{hyperref}
\usepackage{enumitem}
\usepackage{float}
\usepackage{caption}
\usepackage{subcaption}
\usepackage{booktabs} 
\usepackage[numbers,sort&compress]{natbib}
\usepackage{threeparttable}
\usepackage[textwidth=4cm]{todonotes}
\usepackage{comment}
\usepackage{cleveref}

\DeclareMathOperator{\arcsinh}{arcsinh}

\hypersetup{
    colorlinks=true,
    linkcolor=red,
    citecolor=blue,
    urlcolor=blue
}

\theoremstyle{plain}
\newtheorem{theorem}{Theorem}[section]
\newtheorem{lemma}[theorem]{Lemma}
\newtheorem{proposition}[theorem]{Proposition}

\theoremstyle{definition}
\newtheorem{definition}[theorem]{Definition}

\theoremstyle{remark}
\newtheorem{remark}[theorem]{Remark}

\newcommand{\norm}[1]{\left\lVert#1\right\rVert}

\newcommand{\R}{\mathbb{R}}
\newcommand{\C}{\mathbb{C}}

\newcommand{\calI}{\mathcal{I}}

\title{\textbf{Zero-Freeness of the Hard-Core Model with Bounded Connective Constant}}

\author{
Yuan Chen \thanks{School of Computer Science and Technology \& Hefei National Laboratory, University of Science and Technology of China.} \\ {\tt chenyuan@ustc.edu}
\and
Shuai Shao\footnotemark[1]\\
{\tt shao10@ustc.edu.cn}
\and
Ke Shi\footnotemark[1]\\
{\tt self.ke.shi@gmail.com}}
\date{}

\begin{document}

\maketitle
\begin{abstract}
    We study the zero-free regions of the partition function of the hard-core model on finite graphs and their implications for the analyticity of the free energy on infinite lattices. 
    Classically, zero-freeness results have been established up to the tree uniqueness threshold $\lambda_c(\Delta-1)$ determined by the maximum degree $\Delta$. 
    However, for many graph classes, such as regular lattices, the connective constant $\sigma$ provides a more precise measure of structural complexity than the maximum degree. 
    While recent approximation algorithms based on correlation decay~\cite{sinclair2017spatial} and Markov chain Monte Carlo~\cite{efthymiou2026sampling} have successfully exploited the connective constant to improve the threshold to $\lambda_c(\sigma)$, analogous results for complex zero-freeness have been lacking.
    
    In this paper, we bridge this gap by introducing a proper definition of the connective constant for finite graphs based on a lower bound on the number of $k$-depth self-avoiding walks. 
    We prove that for any graph family with a lower connective constant $\mu$, the partition function is zero-free in a complex neighborhood of the interval $[0, \lambda]$ for all $\lambda < \lambda_c(\mu)$. 
    As a direct consequence, we establish the uniqueness and analyticity of the free energy density for infinite lattices up to the connective constant threshold, extending the known regions derived from maximum degree bounds.
    Our proof utilizes a block contraction technique that lifts the correlation decay property from a real interval to a strip-like complex neighborhood. 
\end{abstract}
\thispagestyle{empty}
\clearpage
\newpage
\setcounter{page}{1} 
\section{Introduction}

The \emph{hard-core model} is a fundamental model in statistical physics, dating back to the 1960s \cite{runnels1965hard}. It describes a gas of particles with hard-sphere interactions and is mathematically equivalent to the problem of counting independent sets in a graph when the particle activity $\lambda$ is set to $1$.
Let $G=(V,E)$ be a finite graph, and let $\mathcal{I}(G)$ denote the collection of all independent sets of $G$. 
Given an activity (or fugacity) $\lambda \in\mathbb{C}$, the partition function $Z_G(\lambda)$ of the hard-core model over $G$ is defined as $\sum_{I \in \mathcal I(G)} \lambda^{|I|},$ where each $I\in\calI(G)$ is a configuration associated with the weight $\lambda^{|I|}$.
When $\lambda>0$, this partition function induces the \emph{Gibbs
distribution}
$\mu_G^\lambda(I) \;=\; \frac{\lambda^{|I|}}{Z_G(\lambda)}.$
For a vertex $v\in V$, we denote the marginal probability of $v$ being occupied by $P_{G,v} := \sum_{I \ni v} \mu_G^\lambda(I)$.

The study of the zero-free regions of the partition function $Z_G(\lambda)$ in the complex plane is of profound significance.
In statistical physics, according to the Lee-Yang theory~\cite{lee1952statistical}, the distribution of zeros determines the analyticity of the free energy density $f_G(\lambda)=\log Z_G(\lambda)/|V_G|$, which is a classical notion for defining  (the absence of) phase transitions.
For an infinite lattice $L$, the analyticity of the limiting free energy $f_L(\lambda)$ is typically established by showing that the partition functions of a sequence of finite subgraphs of $L$ are uniformly zero-free in a complex neighborhood of a real interval.  \cite{scott2005repulsive, regts2018zero}.

From a computational perspective, the location of complex zeros is closely related to efficient approximation of the partition functions. 
The study of complex zeros of the partition function $Z_G(\lambda)$ is also of algorithmic interest. 
Using the approach developed by Barvinok \cite{barvinok2016combinatorics} and extended by Patel and Regts \cite{patel2017deterministic},
a zero-free region $\mathbb{C}_\Delta$ of $Z_G(\lambda)$  can be turned into a \emph{fully polynomial-time approximation scheme} (FPTAS) for  computing
$Z_G(\lambda)$ on all  graphs with maximum degree $\Delta$, given $z\in \mathbb{C}_\Delta$.
 An FPTAS is a family of algorithms $\{A_{\varepsilon,\delta}\}$, where  $A_{\varepsilon,\delta}$ is a multiplicative $(1\pm\varepsilon)$-approximation algorithm with running time polynomial in $1/\varepsilon$ and $\log(1/\delta)$ for each $\varepsilon,\delta>0$.

Historically, zero-freeness results for the hard-core model were primarily tied to the maximum degree.
For general finite graphs with maximum degree $\Delta\geq 3$,  Shearer \cite{shearer1985problem} established a zero-free disk around the origin with radius $\lambda_s(\Delta-1)=\frac{(\Delta-1)^{\Delta-1}}{\Delta^\Delta}$ using  Lovász Local Lemma.
A breakthrough~\cite{peters2019conjecture} was the resolution of Sokal's conjecture \cite{scott2005repulsive}, which extended the zero-free region to  a strip-like complex neighborhood of $[0,\lambda_c(\Delta)-\epsilon]$ for arbitrarily small $\epsilon$.
Here $\lambda_c(\Delta-1)=\frac{(\Delta-1)^{\Delta-1}}{(\Delta-2)^\Delta}$  represents the phase transition threshold for the uniqueness of the Gibbs measure on a $\Delta$-regular tree.
More recently,  these zero-freeness results have been refined to a larger complex region~\cite{bencs2023complex} while still within the critical points $(-\lambda_s(\Delta-1), \lambda_c(\Delta-1))$.
Complementary results regarding the existence of zeros near the critical points $-\lambda_s(\Delta-1)$ and $\lambda_c(\Delta-1)$ have also been established ~\cite{peters2019conjecture}.
Stronger zero-freeness results can be obtained when the graph class is restricted to special structures such as claw-free graphs~\cite{heilmann1972theory, chudnovskyseymour07, jerrum2025zero}.

 In statistical physics, however, the focus often shifts to infinite lattices. One can reasonably expect that for such highly structured graphs, the zero-free region should be significantly larger than those implied by maximum degree bounds.
For instance, on the square lattice $\mathbb{Z}^2$,  it has been long conjectured and numerically verified that the free energy is analytic (or equivalently, the partition functions of finite subgraphs are zero-free) on a complex neighborhood of $(0, \lambda^\ast \approx 3.796)$ \cite{vera2013improved}.
This conjectured threshold $\lambda^\ast$ is much larger than the tree threshold $\lambda_c(4)=\frac{(4-1)^{(4-1)}}{(4-2)^4}=1.6875$. 
While rigorous proofs place $\lambda^\ast$ between $2.538$ \cite{sinclair2017spatial} and $5.3506$ \cite{blanca2019phase}, the lower bound of $2.538$ was obtained via analyzing the \emph{correlation decay} property using connective constants, rather than through a direct zero-freeness analysis of the partition functions.

Correlation decay is another classical notion for defining phase transitions, referring to the phenomenon where correlations between vertices decay exponentially with their distance. This property, or more precisely \emph{strong spatial mixing},  can also be exploited to devise FPTASes. 
This method was originally developed by Weitz \cite{weitz2006counting} and independently by Bandyopadhyay and Gamarnik~\cite{bandyopadhyay2008counting} for the hard-core model. 
Weitz’s algorithm approximates the partition function on a graph $G$ by expressing it as a telescoping product of marginal probabilities, utilizing the \emph{self-avoiding walk} (SAW) tree to reduce the computation from $G$ to a tree structure. 
For all graphs of degree at most $\Delta$, SSM holds and FPTASes exist up to the uniqueness threshold $\lambda < \lambda_c(\Delta-1)$. 
A seminal improvement was achieved in \cite{sinclair2017spatial}, which extended the threshold from $\lambda_c(\Delta-1)$ (based on maximum degree) to $\lambda_c(\sigma)$ for graphs with a connective constant at most $\sigma$. 
 For the square lattice, the connective constant is at most $2.64$~\cite{jacobsen2016growth}, which is less than $\Delta -1=3$, yielding a better threshold $\lambda_c(2.64)\approx 2.08$, which was further extended to $2.538$ with additional refinements~\cite{sinclair2017spatial}.

The other major approach is the randomized scheme based on the \emph{Markov Chain Monte Carlo} (MCMC) method, pioneered by Jerrum and Sinclair for the ferromagnetic Ising model \cite{jerrum1993polynomial}.
For the hard-core model on bounded-degree graphs, the mixing time of Glauber dynamics has been extensively studied~\cite{luby1999fast,hayes2006coupling,efthymiou2016convergence}. Using the spectral independence framework \cite{anari2020spectral}, polynomial mixing times were established up to the uniqueness threshold $\Delta-1$ for graphs with degree at most $\Delta$. 
Moreover, the optimal $O(n\log n)$ mixing time  has been proven in a line of recent work~\cite{chen2020optimal,ChenE22,Chen0YZ22}.
For random graphs $G(n,d/n)$, a mixing time of $n^{1+O(1/\log\log n)}$ on typical instances was proved \cite{efthymiou2023mixing}.
There are also results for special graph classes, such as bipartite graphs and claw-free graphs \cite{chen2023uniqueness,chen2024fast,chen2024rapid,jerrum2024glauber,chen2025rapid}.
Notably, for the square lattice, Restrepo et al. \cite{restrepo2013improved} obtained an improved mixing regime for all $0<\lambda<2.3882$ beyond the degree threshold $\lambda_c(3)$.
In a very recent work~\cite{efthymiou2026sampling}, optimal mixing time was obtained for graphs with bounded connective constant $\sigma$, extending the threshold from $\lambda_c(\Delta-1)$ to $\lambda_c(\sigma)$ similar to the improvements seen in correlation decay~\cite{sinclair2017spatial}, though using a different notion of local connective constant.

While both MCMC and correlation decay methods have successfully utilized the connective constant to surpass the uniqueness threshold $\lambda_c(\Delta-1)$, 
zero-freeness results have largely remained tied to the maximum degree $\Delta$. 
This raises a natural question: 
\emph{can a larger zero-free region be established in terms of the connective constant?}

To answer this, we must first define a proper notion of the connective constant for finite graphs. For infinite graphs, widely used definitions capture the exponential growth rate of self-avoiding walks~\cite{hammersley1957percolation, lacoin2014non, grimmett2015bounds}.
However, there are only a few different definitions for connective constant for finite graphs. The definition proposed in \cite{sinclair2014approximation} is  a $\log |V|$-depth connective constant depending on the graph size $V$, while the definition used in \cite{efthymiou2026sampling} is a local connective constant based on fixed-length SAWs.
However, neither definition readily guarantees zero-freeness, as easy counter-examples can be constructed. We introduce a slightly different definition and use it to prove our zero-freeness results.

\subsection*{Our contribution}
We propose the following definition of connective constants.
For a graph $G$ and a vertex $v$, let
$\mathcal{N}_{\leq k}(G,v)$ denote the number of self-avoiding walks of
length no more than $k$ starting at $v$, and let $\mu_{k}(G)=\sup_{v\in V}{\mathcal{N}_{\leq k}(G,v)}^{1/k}$. 
\begin{definition}[Connective constants]
    For a family $\mathcal{H}$ of finite graphs, the $k$-depth connective constant of $\mathcal{H}$ for a fixed $k$, denoted by $\mu_{k}(\mathcal{H})$, is $\sup_{G\in\mathcal{H}}\mu_k(G)$.

     The lower connective constant of $\mathcal{H}$, denoted by $\mu_{\inf}(\mathcal{H})$, is $\inf_{k\geq1}\mu_{k}(\mathcal{H})$.
\end{definition}

For an infinite graph $G$, let $\mathcal{H}_G$ denote the set of all finite subgraphs of $G$. We establish the following relationship: for any infinite graph $G$ with connective constant $\sigma(G)$, $\mu_{\inf}(\mathcal{H}_G) \leq \sigma(G)$. In particular, if $G$ is an infinite lattice, then $\mu_{\inf}(\mathcal{H}_G) = \sigma(G)$. This ensures that we can derive the analyticity of the free energy function on infinite lattices from the zero-freeness of partition functions on their finite subgraphs.

We now state our main zero-freeness result for graphs with bounded lower connective constant. 
 For a real interval $I$ and a constant $\delta>0$, let $\mathcal{L}_\delta(I)=\{z\in \mathbb{C}\mid |z-z_0|\leq \delta, z_0\in I\}$.

\begin{theorem}
\label{thm:zero_free_main}
Let $\mathcal{H}$ be a family of finite graphs with lower connective constant $\mu_{\inf}(\mathcal{H})$.
Then, for any $0<\lambda<\lambda_c(\mu_{\inf}(\mathcal{H}))$, there exists a $\delta$ depending on $\lambda$ and $\mathcal{H}$ such that for  any $H\in \mathcal{H}$, the partition function $Z_H$ is zero-free on $\mathcal{L}_\delta([0, \lambda])$.
\end{theorem}

As a corollary, we have the following result for infinite lattices.

\begin{theorem}
    \label{thm:infinite_main}
    For an infinite lattice $L$ with the connective constant $\sigma(L)$ and any $\lambda < \lambda_c(\sigma(L))$, there exists a $\delta>0$ (depending on $\lambda$ and $L$) such that the free-energy function $f_L$ is unique and analytic on $\mathcal{L}([0, \lambda])$.
\end{theorem}

We provide Table 1 to compare previous thresholds obtained via maximum degree with our new thresholds obtained via the connective constant. 
This table is analogous to Table 1 of ~\cite{sinclair2017spatial} in correlation decay literature, 
but here the thresholds refer to zero-freeness. 

\begin{table}[htbp]
    \centering
    \small 
    \setlength{\tabcolsep}{3pt}     
    
    \label{tab:zerofree_comparison}
    \begin{tabular}{lcccc}
        \toprule
        \textbf{Lattice} & \textbf{Maximum } & \textbf{Degree} & \textbf{Upper bound of } & \textbf{New Threshold} \\
        (Graph Family) & Degree $\Delta$ & Threshold $\lambda_c(\Delta-1)$ & Connective Const. $\sigma$ & $\lambda_c(\sigma)$ \\
        \midrule
        Triangular ($\mathbb{T}$) & 6 & 0.762 & 4.251419~\cite{alm2005upper} & 0.961 \\
        Honeycomb ($\mathbb{H}$) & 3 & 4.000 & 1.847760~\cite{duminil2012connective} & 4.976 \\
        Square ($\mathbb{Z}^2$) & 4 & 1.688 & 2.638158~\cite{jacobsen2016growth} & 2.145($2.538^*$) \\
        Cubic ($\mathbb{Z}^3$) & 6 & 0.762 & 4.7387~\cite{ponitz2000improved} & 0.822 \\
        $\mathbb{Z}^4$ & 8 & 0.490 & 6.8040~\cite{ponitz2000improved} & 0.508 \\
        $\mathbb{Z}^5$ & 10 & 0.360 & 8.8602~\cite{ponitz2000improved} & 0.367 \\
        $\mathbb{Z}^6$ & 12 & 0.285 & 10.8788~\cite{slade1994self} & 0.289 \\
        ($3.12^2$) & 3 & 4.000 & 1.7110412~\cite{jensen1998self} & 6.318 \\
        ($4.8^2$) & 3 & 4.000 & 1.80883001(6)~\cite{jensen1998self} & 5.301 \\
        Kagome & 4 & 1.688 & 2.5606(2)~\cite{jensen1998self} & 2.277 \\
        \bottomrule
    \end{tabular}
    \begin{tablenotes}
        \footnotesize
        \item * This improved bound was originally obtained in \cite{sinclair2017spatial}, and is also achieved in terms of zero-freeness (see Appendix). 
    \end{tablenotes}
    \caption{Thresholds: Max Degree vs. Connective Constant. }
\end{table}

We prove the zero-freeness results following the framework of lifting the contraction property for correlation decay from a real interval to a complex neighborhood~\cite{peters2019conjecture,liu2019fisher,peters2020location,shao2021contraction,liu2022correlation}. 
This advances a recent line of work~\cite{gamarnik2023correlation,regts2023absence,shao2024zero,shao2025zero} establishing bidirectional connections between correlation decay and zero-freeness.
Unlike previous results that analyze the contraction of the 1-layer recurrence function, we consider the block contraction of $k$-layer recurrence functions. 
In addition, we extend the derivation of analyticity of the free energy from zero-freeness on infinite lattices to general infinite graphs in the appendix.

\section{Preliminaries}

In this paper, an infinite graph is a locally finite graph. 

Weitz~\cite{weitz2006counting} introduced the self-avoiding walk (SAW) tree as a tool to analyze the hard-core model.
Given a (finite or infinite) graph $G$ and a root vertex $v$, the SAW tree $T_{\mathrm{SAW}}(G,v)$ is a rooted tree that encodes all self-avoiding walks.

Fix $G=(V,E)$, a root $v\in V$. Given any \emph{ordering scheme} $\{\prec_u: u\in V\}$
where each $\prec_u$ is a total order on $N(u)$,
the SAW tree $T_{\mathrm{SAW}}(G,v)$ is the rooted tree whose nodes are
self-avoiding walks $w=(v=v_0,v_1,\dots,v_t)$ in $G$.
From $w$ we add a child for each extension $(v_0,\dots,v_t,x)$ with
$x\in N(v_t)\setminus\{v_0,\dots,v_t\}$.

If $x\in N(v_t)$ already appeared in $w$ (say $x=v_s$, $s<t$), the extension creates
a \emph{cycle-closing leaf}. Let $y:=v_{s+1}$ be the neighbor used when the walk first
left $x$. If $v_t \prec_x y$, then pin this leaf to \emph{occupied};
otherwise pin it to \emph{unoccupied}.
Weitz proved that the marginal probability at a vertex $v$ in a graph $G$ can be reduced to that on the SAW tree, namely
$P_{G,v} = P_{T_{\mathrm{SAW}}(G,v),\,v}$.
Although the SAW-tree construction may introduce pinned leaves, we can eliminate all pinned vertices by simple reductions. In the hard-core model, a pinned unoccupied leaf can be deleted, while a pinned occupied leaf forces its parent to be unoccupied and thus both the leaf and its parent can be removed.
Consequently, when considering finite graphs, we may without loss of
generality assume that there are no pinned vertices.

For a graph $G$ and a vertex $v$, we give the following notation to count the number of SAWs. 
\begin{itemize}
    \item $\mathcal{N}_k(G,v)$ denotes the number of self-avoiding walks of
length $k$ starting at $v$, and  $\sigma_{k}(G):=\sup_{v\in V}{\mathcal{N}_k(G,v)}^{1/k}$. 
\item $\mathcal{N}_{\leq k}(G,v):=\sum_{i=1}^k\mathcal{N}_i(G,v)$ denotes the number of self-avoiding walks of
length no more than $k$ starting at $v$, and $\mu_{k}(G):=\sup_{v\in V}{\mathcal{N}_{\leq k}(G,v)}^{1/k}$. 
\label{def:2cc}
\end{itemize}

For a graph $G$ and a vertex $v$, let $R_{G,v}=\frac{P_{G,v}}{1-P_{G,v}}$.
Let $T$ be a tree with root $v$ and
children $v_1,\dots,v_d$, and let $T_i$ denote the subtree rooted at $v_i$. 
Then the ratio  $R_{T,v}$ can be computed using $R_{T_i,v_i}$ by the recurrence formula $ R_{T,v} \;=\; \lambda \prod_{i=1}^{d} \frac{1}{1 + R_{T_i, v_i}}.$
We define the one-step recurrence function as
\[
F_{\lambda,d}(x_1,\dots,x_d)\;:=\;\lambda\prod_{i=1}^{d}(1+x_i)^{-1},
\]
so that $R_{T,v}=F_{\lambda,d}(R_{T_1,v_1},\dots,R_{T_d,v_d})$.

For a set $\mathcal{F}=\{f_k\}_{k\ge 1}$ of countably many functions, a point $a$ is an \emph{accumulation point} of the zeros of $\{f_k\}_{k\ge 1}$ if every neighborhood of $a$ contains at least one point $x$ such that $f_k(x)=0$ for some $f_k\in\mathcal{F}$.

\section{Connective constants for finite graphs}

In this section, we first demonstrate that the logarithmic depth definition of the connective constant for finite graphs, as proposed in~\cite{sinclair2017spatial}, fails to ensure zero-freeness.
Subsequently, we analyze the relationship between our newly introduced definition (Definition~\ref{def:2cc}) and the standard connective constant for infinite graphs.
Finally, we establish that zero-freeness for finite subgraphs under our definition implies the analyticity of the free energy density for the corresponding infinite lattice.

\begin{definition}[Log-depth connective constant~\cite{sinclair2017spatial}]
 For a family $\mathcal{H}$ of finite graphs, the log-depth connective constant of $\mathcal{H}$, denoted by $\mu_{\log}(\mathcal{H})$ is at most $d$ if there exist constants $a$ and $c$ such that for any graph $G = (V, E)$ in $\mathcal{H}$ and any vertex $v\in V$, we have $\mathcal{N}_{\leq k}(G,v)\leq c d^\ell$ for all $\ell \geq a \log|V|$.
\end{definition}

We provide the following counter-example to illustrate that a bound on $\mu_{\log}(\mathcal{H})$ does not imply a uniform zero-free region.

\begin{lemma}
\label{lem:log-vacuity}
   Fix $\Delta\geq 3$. Let $T_k$ be the $(\Delta-1)$-ary tree of depth $k$, and $\mathcal{T}=\{T_k\mid k\geq 1\}$.
   Then, for any $d>1$, $\mu_{\log}(\mathcal{T})$ is most $d$, i.e., $\mu_{\log}(\mathcal{T})\leq d$.
\end{lemma}

\begin{proof}
Fix $k$ and write $n:=|V(T_k)|$. On a tree, every self-avoiding walk is a simple path, and each vertex $u\neq v$ at
distance $i$ from $v$ determines a unique self-avoiding walk of length $i=\mathrm{dist}(u,v)$. Hence, for every $r\ge 1$,
\[
(\mu_r(T_k))^r\;=\;\max_{v\in V(T_k)}\mathcal{N}_{\le r}(T_k,v)\;\le\; n-1\;<\; n.
\]
Now take $a:=1/\log d$ and $c:=1$. If $\ell\ge a\log n$, then $d^{\ell}\ge n$ and therefore
$(\mu_\ell(T_k))^\ell\le n\le d^{\ell}$, which is exactly the $\mu_{\log}(\mathcal{T})\leq d$.
\end{proof}

To formalize the inability of the log-depth definition to exclude zeros, we recall a result by Peters and Regts regarding the location of zeros for regular trees.

\begin{proposition}[Proposition~2.1, \cite{peters2019conjecture}]
\label{prop:peters-zeros}
Fix $\Delta\ge 3$. There exists an open set 
\begin{equation}
    U_{\Delta-1}=\left\{-\frac{\alpha(\Delta-1)^{\Delta-1}}{(\Delta-1+\alpha)^\Delta} \;\middle|\; |\alpha|<1\right\}\subset\mathbb{C}
\end{equation} 
such that for all $k\in\mathbb N$ and all
$\lambda\in U_{\Delta-1}$, the hard-core partition function of the depth-$k$ $(\Delta-1)$-ary tree is nonzero. Moreover,
if $\lambda\in\partial U_{\Delta-1}$, then every open neighborhood of $\lambda$ contains a zero of one of these
partition functions.
\end{proposition}

Combining this with Lemma~\ref{lem:log-vacuity}, we establish the following non-zero-freeness result.

\begin{lemma}
    There exists a family $\mathcal{H}$ of finite graphs with $\mu_{\log}(\mathcal{H})\leq d$ such that there exists an accumulation point of zeros of partition functions over $\mathcal{H}$ for some $\lambda\in(0, \lambda_c(d)).$ In other words, one cannot find a strip-like complex zero-free region in terms of $\mu_{\log}(\mathcal{H})$.
\end{lemma}

\begin{proof}
    Let $\mathcal{H}$ be the $(\Delta-1)$-ary tree. By Proposition~\ref{prop:peters-zeros}, taking $\alpha=-1$ yields a point

    $$\lambda=\frac{(\Delta-1)^{\Delta-1}}{(\Delta-2)^\Delta}$$

    lying on \(\partial U_{\Delta-1}\) and on the positive real axis.

    As in Lemma~\ref{lem:log-vacuity}, for any $d>1$ we can ensure $\mu_{\log}(\mathcal{T})\le d$ by choosing $a$ sufficiently large. Since $\lambda_c(\cdot)$ is decreasing, we may choose $d$ such that $\lambda_c(d)>\lambda$. Consequently, the partition function has at least one zero in every open complex neighborhood of $[0,\lambda_c(d))$, completing the proof.
\end{proof}

We now examine the relationship between our definition of connective constants (Definition~\ref{def:2cc}) and the standard definition for infinite graphs.

\begin{definition}[Connective constants for infinite graphs]
    \label{pro:cc}
    For an infinite graph $G$, the connective constant $\sigma(G)$ of $G$ is $\limsup_{k\rightarrow\infty}\sigma_k(G)$.
\end{definition}

\begin{proposition}
Let $G$ be a vertex-transitive infinite graph. Then $\sigma(G)=\lim_{k\rightarrow\infty}\sigma_k(G)=\lim_{k\rightarrow\infty}\mu_k(G)$.
\end{proposition}

\begin{proof}
    For a vertex-transitive infinite graph $G$, from \cite{grimmett2015bounds} we know that $\sigma(G)=\lim_{k\rightarrow\infty}\sigma_k(G)=\inf_{k\ge 1}\sigma_k(G)$.
    For any vertex $v\in V(G)$, $$\liminf_{k\to\infty}(\mathcal{N}_{\leq k}(G,v))^{1/k}\ge\inf_{k\ge1}(\mathcal{N}_{k}(G,v))^{1/k}=\sigma(G).$$ Fix any $\varepsilon>0$. There exists $K>0$ such that for all $k\ge K$,
    $$\mathcal{N}_{\leq k}(G,v)\le\mathcal{N}_{\leq K-1}(G,v)+\sum^k_{i=K}(\sigma(G)+\varepsilon)^i\le M(\sigma(G)+\varepsilon)^k$$
    where $M<\infty$ is a constant. Let $\varepsilon\downarrow0$; it follows that $\limsup_{k\to\infty}(\mathcal{N}_{\leq k}(G,v))^{1/k}\le\sigma(G)$. Taking the supremum over $v$ and the infimum over $k$ yields $\lim_{k\rightarrow\infty}\mu_k(G)=\sigma(G)$.
\end{proof}

For an infinite graph $G$, let $\mathcal{H}_G$ denote the set of all finite subgraphs of $G$. We have the following relation between the lower connective constant of the finite subgraphs and the connective constant of the infinite graph.

\begin{lemma}
    \label{lem:muinf_le_sigma}
    For any infinite graph $G$, $\mu_{\inf}(\mathcal{H}_G)\leq \sigma(G).$
    In particular, if $G$ is an infinite lattice, then $\mu_{\inf}(\mathcal{H}_G)= \sigma(G).$
\end{lemma}

\begin{proof}
    We defer the proof to the Appendix \ref{sec:infinite_lattice_free_energy}.
\end{proof}

The following lemma links the analyticity of the free energy on an infinite lattice to the zero-freeness of partition functions on its finite subgraphs.

\begin{lemma}
    \label{lem:infinite_main}
    For an infinite lattice $L$, if there is a complex region $R=\mathcal L ([0, \lambda^\ast])$ such that for every $H\in \mathcal{H}_L$, the partition function $Z_H$ is zero-free on $R$, then the free energy function $f_L$ is unique and analytic on $R$. 
\end{lemma}

\begin{proof}
    When $L$ is a square lattice $\mathbb{Z}^d$, the proof is given in \cite{scott2005repulsive}, and it can be extended to arbitrary lattices.
    In the appendix, we actually prove this lemma for a broader class of infinite graphs satisfying certain conditions, of which lattices are special cases.  
\end{proof}

We are now ready to prove our main result  regarding the analyticity of the free energy density on infinite lattices.

\begin{theorem}[Theorem~\ref{thm:infinite_main} restated]
    For an infinite lattice $L$ with the connective constant $\sigma(L)$ and any $\lambda < \lambda_c(\sigma(L))$, there exists a $\delta>0$ (depending on $\lambda$ and $L$) such that the free-energy function $f_L$ is unique and analytic on $\mathcal{L}([0, \lambda])$.
\end{theorem}

\begin{proof}
    For a lattice $L$, we have $\sigma(L)=\mu_{\inf}(\mathcal{H}_L)$ by Lemma~\ref{lem:muinf_le_sigma}. Hence, by Theorem~\ref{thm:zero_free_main}, for any $0<\lambda<\lambda_c(\sigma(L))=\lambda_c(\mu_{\inf}(\mathcal{H}_L))$, the partition function $Z_{H}$ is zero-free on $\mathcal{L}_{\delta}([0,\lambda])$ for every $H\in \mathcal{H}_L$, where $\delta$ depends only on $\lambda$ and $\mathcal{H}_L$. Invoking Lemma~\ref{lem:infinite_main}, we conclude that the free energy exists uniquely and is analytic on $\mathcal{L}_{\delta}([0,\lambda])$, which completes the proof.   
\end{proof}

\section{Zero-freeness under $k$-depth connective constant}\label{sec:positive}
In this section, we prove \Cref{thm:zero_free_main}. 
It suffices to prove the following zero-freeness results in terms of fixed-depth connective constants. 

\begin{theorem}[Main]\label{thm:zero-free}
Fix an integer $k\ge 1$ and real number $d > 1$.
Let $\mathcal{G}$ be the family of finite graphs with bounded $k$-depth connective constant
$\mu_k(\mathcal{G})\le d$.
There exists a complex region $\mathcal{U}\subseteq \C$ that contains the real interval $[0,\lambda_c(d))$ such that for every $\lambda\in \mathcal{U}$ and every $G\in\mathcal{G}$,
$Z_G(\lambda) \neq 0$.
\end{theorem}

We prove this result following the framework of contraction implying zero-freeness.
Unlike previous results analyzing the contraction of the 1-layer recurrence function $R$, 
we consider the block contraction of $k$-layer recurrence functions.

\subsection{Depth-$k$ cutsets and $k$-layer block operators}
The hard-core recursion is a \emph{one-step} map from the occupation ratios of a vertex's children to the occupation ratio
of the vertex itself. In our arguments it is convenient to bundle \emph{$k$ successive steps} into a single operator
that maps the ``interface'' at depth $k$ to the root.
This viewpoint is especially natural for Weitz's SAW tree: some branches terminate early due to boundary conditions
(pinned leaves), while the vertices at depth exactly $k$ should be regarded as \emph{free} interface nodes---beyond
which the tree may continue (and will be handled by subsequent blocks).

\begin{definition}[Depth-$k$ cutset $C_k$]\label{def:cutset_Ck}
Let $T$ be a rooted tree with root $r$. For an integer $k\ge 0$, define the \emph{depth-$k$ cutset}
\[
C_k(T) \;:=\;\{v\in V(T): \mathrm{dist}_T(r,v)=k\}.
\]
Equivalently, $C_k(T)$ is the vertex set on level $k$ of $T$.
In particular, every (simple) path starting at $r$ and going to any vertex at distance $>k$ must intersect $C_k(T)$.
When the underlying tree is clear from context, we write $C_k$ for $C_k(T)$.
\end{definition}

\begin{definition}[$k$-layer block operator]\label{def:block_function}
Fix $\lambda \in \R$ and an integer $k\ge 1$. Let $T$ be a rooted tree with root $r$.
The \emph{$k$-layer block operator} is the map
\[
\mathcal{F}_{\lambda,T,k}:\ {\R}^{\,C_k(T)}\to {\R}
\]
defined by evaluating the hard-core ratio recursion on the tree depth-$k$
interface: given an input vector $\mathbf{R}\in{\R}^{\,C_k(T)}$, we assign ratios on the lefts of $T$ by
\[
R_u :=
\begin{cases}
(\mathbf{R})_u, & u\in C_k(T),\\
\lambda, & \text{otherwise},
\end{cases}
\]
and then propagate upward: for every internal vertex $v$ of $T$ with children $v_1,\dots,v_d$ (in $T$), set
\[
R_v \;:=\; \lambda\prod_{i=1}^{d}(1+R_{v_i})^{-1}.
\]
Finally define $\mathcal{F}_{\lambda,T,k}(\mathbf{R}):=R_r$.
We use $\mathcal{F}_{\lambda,T,k}^\varphi(\mathbf{x})$ to denote the composition
\[
\mathcal{F}_{\lambda,T,k}^\varphi(\mathbf{x})
=
\varphi \circ \mathcal{F}_{\lambda,T,k} \circ \varphi^{-1}(\mathbf{x})
=
\varphi\!\left(
\mathcal{F}_{\lambda,T,k}\bigl(\varphi^{-1}(x_1), \varphi^{-1}(x_2), \dots\bigr)
\right).
\]

The key point is that the only \emph{free} inputs are the vertices in $C_k(T)$: $\mathcal{F}_{\lambda,T,k}$ is an interface map
from the $k$th layer to the root.

\end{definition}

\subsection{Real block contraction}\label{subsec:block_contraction}

We work with a $k$-layer \emph{block} recursion (built on the SAW tree) rather than a single-step recurrence.
Accordingly, we introduce a contraction notion stated directly for our block operator, following the real-contraction
framework of Shao--Sun~\cite{shao2021contraction}.

\begin{definition}[Real block contraction]\label{real_contraction}
Fix $k\ge 1$ and $d>1$.
We say that $\lambda \in \R$ satisfies \emph{block real contraction} 
for $k$ and $d$ if there exists a compact real interval
$\mathcal{J}\subseteq\R$ with $\lambda\in\mathcal{J}$ and $-1 \notin \mathcal{J}$,
and a real-analytic function
$\varphi:\mathcal{J}\to\mathcal{I}$ with $\varphi'(x) \neq 0$
for all $x\in \mathcal{J}$, such that the following conditions hold.
Let $\mathcal{T}$ denote the set of trees $T$ with depth at most $k$
and connective constant $\mu_k(T)\le d$.
\begin{enumerate}[leftmargin=2em]
    \item $\mathcal{F}_{\lambda,T,k}\bigl(\mathcal{J}^{C_k(T)}\bigr) \subseteq \mathcal{J}$
    for all $T\in \mathcal{T}$.
    
    \item There exists $\eta>0$ such that
    \(
        \bigl\|\nabla \mathcal{F}^{\varphi}_{\lambda,T}(\mathbf x)\bigr\|_{1}
        \le 1-\eta
    \)
    for all $\mathbf x\in \mathcal{I}^{C_k(T)}$ and all $T\in \mathcal{T}$.
\end{enumerate}
\end{definition}

\begin{remark}
Definition~\ref{real_contraction} is a direct generalization of the \emph{real contraction} framework of
Shao--Sun~\cite{shao2021contraction}
to the particular block function $\mathcal{F}^{\varphi}_{T,k}$ we analyze in this paper.
In particular, once we verify the two points above for our block operator (uniformly over rooted trees that arise from the SAW tree
construction), the real-to-complex extension theorem of~\cite{shao2021contraction} applies.
\end{remark}

\begin{lemma}[Real block contraction]\label{lem:real-contraction}
    Fix $k\geq 1$ and $d>1$, $\lambda \in [0,\lambda_c(d))$ 
    satisfies real block contraction for $k$ and $d$.
\end{lemma}

To prove real block contraction, we utilize the tools developed to establish strong spatial mixing in \cite{sinclair2017spatial}. Applying \cite[Lem.~3.3]{sinclair2017spatial} to the hard-core model (see \cite[Sec.~4]{sinclair2017spatial}), we obtain the following one-step decay lemma
with potential function $\varphi(x)=\arcsinh(\sqrt{x})$.

\begin{lemma}
\label{lem:one_step_decay}
Fix $d>1$ and $\lambda \in [0,\lambda_c(d))$.
There exist positive constants $a,q$ with $\frac{1}{a}+\frac{1}{q}=1$ and
a constant $\alpha=\alpha(\lambda)<1/d$ such that for any two vectors
$\mathbf{x},\mathbf{y}\in \varphi([0,\lambda])^{d}$,
\[
    \bigl|F_{\lambda,d}^{\varphi}(\mathbf{x}) - F_{\lambda,d}^{\varphi}(\mathbf{y})\bigr|^{q}
    \;\le\;
    \alpha \sum_{i=1}^{d} |x_i-y_i|^{q}.
\]
\end{lemma}

We now amplify one-step decay into a $k$-step Lipschitz bound for the block operator.

\begin{lemma}
\label{lem:block_lipschitz}
Let $\varphi(x) = \arcsinh (\sqrt{x})$ and let $\alpha(\lambda)$ be as above.
Fix any rooted tree $T$ together with a leaf pinning as in Definition~\ref{def:block_function}. Then for any depth
$k \ge 1$ and any $\mathbf{x},\mathbf{y}\in \varphi([0,\lambda])^{C_k(T)}$,
\begin{equation}
|\mathcal{F}^{\varphi}_{\lambda, T,k}(\mathbf{x}) - \mathcal{F}^{\varphi}_{\lambda, T,k}(\mathbf{y})|^q
\;\le\;
\alpha(\lambda)^k \sum_{u \in C_k(T)} |x_u - y_u|^q.
\end{equation}
\end{lemma}

\begin{proof}
We proceed by induction on $k$.

\noindent\textbf{Base case ($k=1$).}
This is Lemma~\ref{lem:one_step_decay}.

\noindent\textbf{Inductive step.}
Assume the bound holds for $k-1$. Let $v$ be the root of $T$, with children $u_1,\dots,u_{d_v}$.
Applying Lemma~\ref{lem:one_step_decay} at $v$ yields
\[
\bigl|\mathcal{F}^{\varphi}_{\lambda, T,k}(\mathbf{x}) - \mathcal{F}^{\varphi}_{\lambda, T,k}(\mathbf{y})\bigr|^q
\;\le\;
\alpha(\lambda)\sum_{i=1}^{d_v}\Bigl|
\mathcal{F}^{\varphi}_{\lambda, T_{u_i},k-1}(\mathbf{x}_{u_i}) -
\mathcal{F}^{\varphi}_{\lambda, T_{u_i},k-1}(\mathbf{y}_{u_i})
\Bigr|^q,
\]
where $\mathbf{x}_{u_i}$ and $\mathbf{y}_{u_i}$ are the restrictions to the depth-$(k-1)$ cutset of the subtree rooted at $u_i$.
By the inductive hypothesis, each term is at most
$\alpha(\lambda)^{k-1}\sum_{w\in C_{k-1}(T_{u_i})}$ $|x_w-y_w|^q$.
Summing over children gives
\[
\bigl|\mathcal{F}^{\varphi}_{\lambda, T,k}(\mathbf{x}) - \mathcal{F}^{\varphi}_{\lambda, T,k}(\mathbf{y})\bigr|^q
\;\le\;
\alpha(\lambda)^k \sum_{z\in C_k(T)} |x_z-y_z|^q,
\]
since $C_k(T)$ decomposes as the disjoint union of the sets $C_{k-1}(T_{u_i})$ over the children $u_i$.
\end{proof}

We are ready to prove the real block contraction.

\begin{proof}[Proof of \Cref{lem:real-contraction}]
    Fix $\lambda \in [0,\lambda_c(d))$. Let $\mathcal{J}=[0,\lambda]$ and $\varphi(x) = \arcsin(\sqrt{x})$,
We verify the two conditions in Definition~\ref{real_contraction}.

    Recall the one-step recursion formula, $R_v = \lambda \prod_{i=1}^{d_v} (1+R_{v_i})^{-1} \in [0,\lambda]$ holds trivially for any vertex $v$ whose children $v_1,\dots,v_{d_v}$ have ratios in $[0,\lambda]$, thus holds for the root by 
    backward induction from the leaves. Thus 
 $\mathcal{F}_{\lambda,T,k}\bigl(\mathcal{J}^{C_k(T)}\bigr) \subseteq \mathcal{J}$
 holds.

Let $m:=|C_k(T)| \leq d^k$. By \Cref{lem:block_lipschitz} and 
Hölder's inequality, the Lipschitz constant of $q$-norm
implies the dual gradient bound
\[
\sup_{\mathbf{x}\in \mathcal{I}^{C_k(T)}}\ \norm{\nabla \mathcal{F}^{\varphi}_{\lambda,T,k}(\mathbf{x})}_a
\;\le\;
\alpha(\lambda)^{k/q}.
\]
For any vector $g\in\mathbb{R}^{m}$, the norm comparison inequality gives
$\|g\|_1 \le m^{1-1/a}\|g\|_a = m^{1/q}\|g\|_a$.
Hence
\[
\sup_{\mathbf{x}\in \mathcal{I}^{C_k(T)}}\ \norm{\nabla \mathcal{F}^{\varphi}_{\lambda,T,k}(\mathbf{x})}_1
\;\le\;
m^{1/q}\alpha(\lambda)^{k/q}
\;\le\;
(d^k)^{1/q}\alpha(\lambda)^{k/q}
=
\bigl(d\,\alpha(\lambda)\bigr)^{k/q}.
\]
Since $\alpha(\lambda)<1/d$, we have $d\,\alpha(\lambda)<1$.
Take
\[
\eta := 1-\bigl(d\,\alpha(\lambda)\bigr)^{k/q}>0.\qedhere
\]
\end{proof}

\subsection{Extension to complex block contraction}

\begin{definition}[Complex block contraction]
\label{def:complex_contraction}
Fix $k\ge 1$ and $d>1$.
We say that $\lambda \in \C$ satisfies \emph{complex block contraction}
for $k$ and $d$ if there exists a closed and bounded complex region
$\mathcal{Q}\subseteq \C$ with $\lambda\in\mathcal{Q}$ and $-1 \notin \mathcal{Q}$,
and an analytic function $\varphi:\mathcal{Q}\to\mathcal{P}$ such that
$\varphi^{-1}$ is also analytic and $\mathcal{P}$ is convex, such that the
following conditions hold.
Let $\mathcal{T}$ denote the set of trees $T$ with depth at most $k$
and connective constant $\mu_k(T)\le d$.
\begin{enumerate}[leftmargin=2em]
    \item $\mathcal{F}_{\lambda,T,k}(\mathcal{Q}^{C_k(T)}) \subseteq \mathcal{Q}$
    for all $T\in \mathcal{T}$.
    
    \item There exists $\eta>0$ such that
    $\bigl\|\nabla \mathcal{F}^{\varphi}_{\lambda,T}(\mathbf z)\bigr\|_{1}
    \le 1-\eta$
    for all $\mathbf z\in \mathcal{P}^{C_k(T)}$ and all $T\in \mathcal{T}$.
\end{enumerate}
\end{definition}

Exactly as in \cite{shao2021contraction}, we may now invoke a real-to-complex extension theorem to transfer
real block contraction to complex block contraction.

\begin{lemma}[{\cite[Theorem~4.4]{shao2021contraction}}]
\label{lem:real_to_complex_extension}
If $\lambda_0\in\R$ satisfies block real contraction for $d$ and $k$, then there exists $\delta>0$ such that for
every $\lambda\in\C$ with $|\lambda-\lambda_0|<\delta$,     $\lambda$ satisfies complex block contraction for $d$ and $k$.
\end{lemma}

\begin{remark}
    Even though Theorem~4.4 in \cite{shao2021contraction} is stated for one-step recursions, the same proof applies verbatim to our $k$-layer block operator $\mathcal{F}^{\varphi}_{\lambda,T,k}$ due to the trees with depth at most $k$ and
    $\mu_k(T)\le d$ are finite in number.
\end{remark}

Now, we are ready to prove \Cref{thm:zero-free}.

    

\begin{proof}
By \Cref{lem:real-contraction,lem:real_to_complex_extension}, there exists a complex region
$\mathcal{U}$ containing $[0,\lambda_c(d))$ such that every $\lambda\in\mathcal{U}$
satisfies complex block contraction for $k$ and $d$.

Let $G=(V,E)\in\mathcal{G}$. We prove that $Z_G(\lambda)\neq 0$ by induction on $|V|$.

\smallskip
\noindent\textbf{Base case.}
If $|V|=1$, then for the unique vertex $v$ we have $R_{G,v}(\lambda)=\lambda\in\mathcal{Q}$.

\smallskip
\noindent\textbf{Inductive step.}
We show that $R_{G,v}(\lambda)\in\mathcal{Q}$ for every $G\in\mathcal{G}$ and every vertex $v\in V$.
Since $-1\notin\mathcal{Q}$, this implies $R_{G,v}(\lambda)\neq -1$ and hence $Z_G(\lambda)\neq 0$.

Fix an integer $t\ge 2$ and assume that $Z_H(\lambda)\neq 0$ for all graphs $H\in\mathcal{G}$ with $|V(H)|\le t-1$.
Let $G=(V,E)\in\mathcal{G}$ with $|V|=t$ and fix $v\in V$.
Recall that $R_{G,v}(\lambda)$ can be expressed as the ratio of partition functions with $v$ occupied and unoccupied,
namely $R_{G,v}(\lambda)=Z^{\mathrm{occ}}_{G,v}(\lambda)/Z^{\mathrm{unocc}}_{G,v}(\lambda)$.
For the hard-core model, $Z^{\mathrm{unocc}}_{G,v}(\lambda)=Z_{G-v}(\lambda)$.
By the induction hypothesis, $Z_{G-v}(\lambda)\neq 0$, so $R_{G,v}(\lambda)$ is well-defined.

Moreover, $R_{G,v}(\lambda)$ can be computed on the self-avoiding-walk tree $T=T_{\mathrm{SAW}}(G,v)$, i.e.,
$R_{G,v}(\lambda)=R_{T,v}(\lambda)$.
Let $C_k(T)=\{v_1,v_2,\dots,v_s\}$ be the $k$-depth cutset, and let $T_{v_i}$ denote the subtree of $T$ rooted at $v_i$.
By the definition of the SAW tree, we have $\mu_k(T)\le d$ whenever $\mu_k(G)\le d$. We truncate $T$ at depth $k$, and denote the resulting tree by $T^{[k]}$.
Then
\[
R_{T,v}(\lambda)
= \mathcal{F}_{\lambda,T^{[k]},k}\bigl(R_{T_{v_1},v_1}(\lambda),R_{T_{v_2},v_2}(\lambda),\dots,R_{T_{v_s},v_s}(\lambda)\bigr).
\]
Each subtree $T_{v_i}$ has strictly fewer vertices than $T$, and hence corresponds to an instance with fewer than $t$ vertices;
therefore, by the induction hypothesis we have $R_{T_{v_i},v_i}(\lambda)\in\mathcal{Q}$ for all $i=1,\dots,s$.
By the complex block contraction, it follows that $R_{T,v}(\lambda)\in\mathcal{Q}$.
Consequently $R_{G,v}(\lambda)\in\mathcal{Q}$ as well, completing the induction and proving $Z_G(\lambda)\neq 0$.
\end{proof}

\appendix

\section{From finite sub-graphs to infinite graphs: normalized free energy}
\label{sec:infinite_lattice_free_energy}

In this section, we lift our finite-volume zero-freeness results to the existence, analyticity, and uniqueness of the normalized free energy in the infinite-volume limit—one of the central thermodynamic quantities. Our arguments follow the framework of \cite{scott2005repulsive}. In addition, we introduce Assumption $\mathcal{A}$, which allows this finite-to-infinite upgrade to apply to a broader class of infinite graphs. Moreover, by exploiting the special structure of graphs satisfying Assumption $\mathcal{A}$, we can sharpen the result by pushing the unique analytic region up to the threshold governed by the connective constant of the infinite graph.

\subsection{Normalized free energy on infinite graphs}
Let $G=(V,E)$ be a (typically infinite) locally finite graph, and let $H\subseteq G$ be a finite induced sub-graph.
For the hard-core model with complex activity $z\in\mathbb{C}$, write $Z_H(z)$ for the partition function on $H$.
We study thermodynamic behavior via the \emph{normalized free energy}
\begin{equation}
  f_H(z)\;:=\;\frac{1}{|V(H)|}\log Z_H(z),
\end{equation}
where $\log$ denotes a branch of the logarithm chosen on a domain where $Z_H$ has no zeros. For infinite graph, free energy is defined along an exhaustion sequence $H_n\uparrow G$ (e.g.\ a F\o lner/van~Hove sequence) by
\begin{equation}
  f_G(z)\;:=\;\lim_{n\to\infty}\frac{1}{|V(H_n)|}\log Z_{H_n}(z),
\end{equation}

Which implies the unique thermodynamic limit exists. 

\begin{theorem}
    Let $G$ be a local finite infinite graph, which holds the Assumption $\mathcal{A}$(we will introduce later). Fix any $\lambda<\lambda_c(\sigma(G))$ , there exist a $\delta>0$(depend on $\lambda$ and $G$), then the normalized free-energy is analytic and unique on $\mathcal{L}_{\delta}([0, \lambda])$.
    \label{thm:infinite_threshold_strong}
\end{theorem}

\begin{proof}
    we will prove in the following subsections.
\end{proof}

The theorem \ref{thm:infinite_threshold_strong} is actually a stronger result than theorem \ref{thm:infinite_main}.

\subsection{Infinite periodic graphs and Assumption $\mathcal{A}$}

In this subsection, we will introduce the basic settings and properties of infinite graph and  exhaustion sub-graph sequence. Further more we will state the Assumption $\mathcal{A}$ that is important to the proof of our main result.
\begin{lemma}
Let $G=(V,E)$ be a (countable) graph and fix $\lambda\ge 0$.
For every finite $\Lambda\subset V$, define the hard-core partition function
\[
Z_\Lambda(\lambda)\;=\;\sum_{\substack{I\subseteq \Lambda\\ I\ \text{independent in }G}}
\lambda^{|I|}.
\]
Suppose that $\Lambda$ is decomposed as a disjoint union
\[
\Lambda \;=\; \Big(\bigcup_{i=1}^n \Lambda_i\Big)\;\cup\;\Lambda_0 .
\]
Then
\[
Z_\Lambda(\lambda)\;\le\;\Big(\prod_{i=1}^n Z_{\Lambda_i}(\lambda)\Big)\,(1+\lambda)^{|\Lambda_0|}.
\]
Moreover, if the blocks are pairwise non-adjacent, i.e.
\[
E(\Lambda_i,\Lambda_j)=\varnothing \qquad \text{for all } 1\le i<j\le n,
\]
then
\[
Z_\Lambda(\lambda)\;\ge\;\prod_{i=1}^n Z_{\Lambda_i}(\lambda).
\]
\end{lemma}

\begin{proof}

\emph{Upper bound.}

    For any independent set $I\subseteq \Lambda$, since $\Lambda$ being divide into disjoint $\Lambda_0, \Lambda_1, ..., \Lambda_n$
    $$I=(I\cap\Lambda_0)\cup(I\cap\Lambda_1)\cup\dots\cup(I\cap\Lambda_n)$$
    Let $I_i$ denote $I\cap\Lambda_i$, we have
    $$I\in\calI (\Lambda)\Rightarrow I_i\in\calI (\Lambda_i)$$
    $$\calI(\Lambda)\subseteq\calI(\Lambda_1)\times\calI(\Lambda_2)\times...\times\calI(\Lambda_n)\times\mathcal{P}(\Lambda_0)$$
    where $\mathcal{P}(\Lambda_0)$ is the power set of $\calI(\Lambda_0)$, it turns out 
    \begin{align}
       Z_\Lambda(\lambda)&=\sum_{I\in\calI(\Lambda)}\lambda^{|I|}\le\sum_{I_1\in\calI(\Lambda_1)}\cdots\sum_{I_n\in\calI(\Lambda_n)}\sum_{S\subseteq\Lambda_0}\lambda^{|S|}\prod^n_{i=1}\lambda^{|I_i|}\\
       &=\bigl(\prod^n_{i=1}\sum_{I_i\in\calI(\Lambda_i)}\lambda^{|I_i|}\bigr)(\sum_{S\subseteq\Lambda_0}\lambda^{|S|})\\
       &=\Big(\prod_{i=1}^n Z_{\Lambda_i}(\lambda)\Big)\,(1+\lambda)^{|\Lambda_0|}
    \end{align}

\emph{Lower bound.}
    Since we choose each block $\Lambda_i(i\ge 1)$ is non-adjacent, there are 
    \begin{align}
       Z_\Lambda(\lambda)&=\sum_{I\in\calI(\Lambda)}\lambda^{|I|}\ge\sum_{I_1\in\calI(\Lambda_1)}\cdots\sum_{I_n\in\calI(\Lambda_n)}\prod^n_{i=1}\lambda^{|I_i|}\\
       &=\prod_{i=1}^n Z_{\Lambda_i}(\lambda)
    \end{align}
    
\end{proof}

\begin{definition}[F\o lner sequence]
Let $G=(V,E)$ be a locally finite infinite graph. For a finite set $S\subset V$, define its (vertex) boundary
\[
\partial S \;:=\; \{\, v\in S : \exists\, u\in V\setminus S \text{ with } \{u,v\}\in E \,\}.
\]
A sequence of finite sets $(S_n)_{n\ge 1}$ with $|S_n|\to\infty$ is called a \emph{F\o lner sequence} if
\[
\frac{|\partial S_n|}{|S_n|}\xrightarrow[n\to\infty]{}0.
\]
\end{definition}

\begin{definition}[van Hove (F\o lner--van Hove) sequence]
Let $G=(V,E)$ be a locally finite infinite graph and let $\mathrm{dist}(\cdot,\cdot)$ be the graph distance.
For $r\in\mathbb{N}$ and a finite set $S\subset V$, define the $r$-thick boundary
\[
\partial^{(r)} S \;:=\; \bigl\{\, v\in V : \mathrm{dist}(v,S)\le r \ \text{and}\ \mathrm{dist}(v,V\setminus S)\le r \,\bigr\}.
\]
A sequence of finite sets $(S_n)_{n\ge 1}$ with $|S_n|\to\infty$ is called a \emph{van Hove sequence}
(or \emph{F\o lner--van Hove sequence}) if for every fixed $r\in\mathbb{N}$,
\[
\frac{|\partial^{(r)} S_n|}{|S_n|}\xrightarrow[n\to\infty]{}0.
\]
\end{definition}

\begin{definition}[Assumption $\mathcal{A}$]\label{def:assumptionA}
Let $G=(V,E)$ be a connected, locally finite, infinite graph.
We say that $G$ satisfies \emph{Assumption $\mathcal{A}$} if there exist an integer $d\ge 1$
and a subgroup $\Gamma\le \mathrm{Aut}(G)$ such that the following hold:

\begin{enumerate}

\item \textbf{Translation group.}
There exist an integer $d\ge 1$ and automorphisms $\tau_1,\dots,\tau_d\in \mathrm{Aut}(G)$ such that
(i) $\tau_i\tau_j=\tau_j\tau_i$ for all $i,j$, and
(ii) the map $\mathbb{Z}^d\to\Gamma$ given by
\[
(n_1,\dots,n_d)\longmapsto \tau_1^{n_1}\cdots \tau_d^{n_d}
\]
is a group isomorphism. Equivalently, $\Gamma=\langle\tau_1,\dots,\tau_d\rangle=\{\tau_1^{n_1}\cdots\tau_d^{n_d}:(n_1,\dots,n_d)\in\mathbb{Z}^d\}$ is a free abelian group of rank $d$
(i.e., isomorphic to $\mathbb{Z}^d$).

\item \textbf{Free.}
The action of $\Gamma$ on $V$ is free, i.e.
\[
\forall v\in V,\quad \bigl(\gamma v = v \ \Rightarrow\  \gamma = e\bigr),
\]
where $e$ denotes the identity element of $\Gamma$.

\item \textbf{Cocompact.}
The action of $\Gamma$ on $V$ has finitely many vertex-orbits, i.e. there exists a finite set
$F\subset V$ such that
\[
V = \Gamma F := \{\gamma f:\ \gamma\in\Gamma,\ f\in F\}.
\]
Any such finite set $F$ is called a \emph{(finite) fundamental domain}.

\end{enumerate}
\end{definition}

\begin{remark}
Assumption $\mathcal{A}$ is satisfied by many infinite graphs. 
In particular, for all periodic graph with finite unit cell(i.e.~lattice), Assumption $\mathcal{A}$ holds.
\end{remark}

\begin{lemma}[Assumption $\mathcal{A}$ implies existence of a F\o lner-van Hove sequence]\label{lem:P_implies_vanhove}
Let $G=(V,E)$ be a connected, locally finite, infinite graph. Assume that $G$ satisfies
Assumption $\mathcal{A}$ from Definition~\ref{def:assumptionA}: namely, there exist an integer $d\ge 1$
and a subgroup $\Gamma\le \mathrm{Aut}(G)$ such that $\Gamma\cong \mathbb{Z}^d$, the action of $\Gamma$
on $V$ is free, and there exists a finite fundamental domain $F\subset V$ with $V=\Gamma F$.
Then $G$ admits a F\o lner-van Hove sequence $(\Lambda_n)_{n\ge 1}$ in the following:
for every fixed $r\in\mathbb{N}$,
\[
\frac{|\partial^{(r)}\Lambda_n|}{|\Lambda_n|}\longrightarrow 0 \qquad (n\to\infty),
\]
where $\partial^{(r)}\Lambda:=\{v\in \Lambda:\ \mathrm{dist}_G(v,V\setminus \Lambda)\le r\}$.
\end{lemma}

\begin{proof}
\textbf{Step 1: Reduce to a fundamental domain consisting of orbit representatives.}
Since $V=\Gamma F$ and $F$ is finite, $V$ is a union of finite number of $\Gamma$-orbits.
Choose a subset $F_0\subseteq F$ containing exactly one representative from each $\Gamma$-orbit.
Then $F_0$ is finite and
\[
V=\bigcup_{f\in F_0}\Gamma f.
\]
Replacing $F$ by $F_0$, we may assume without loss of generality that $F$ consists of orbit
representatives.

\smallskip
\textbf{Step 2: Build a cell-aligned sequence of finite sets.}
Fix an isomorphism $\Gamma\cong \mathbb{Z}^d$ and identify elements $\gamma\in\Gamma$ with vectors
$\mathbf{n}=(n_1,\dots,n_d)\in\mathbb{Z}^d$. For $n\ge 1$ define the box
\[
B_n:=\{-n,-n+1,\dots,n\}^d\subset \Gamma
\]
and the corresponding finite vertex set
\[
\Lambda_n \;:=\; \bigcup_{\gamma\in B_n} \gamma F \;\subset\; V.
\]
Since the sets $\gamma F$ are pairwise disjoint (because $F$ contains orbit representatives and
the action is free), we have
\[
|\Lambda_n| \;=\; |F|\,|B_n| \;=\; |F|\,(2n+1)^d \xrightarrow[n\to\infty]{}\infty.
\]

\smallskip
\textbf{Step 3: Finite range of ``cell displacements'' across edges.}
Define the finite set of displacements
\[
D\;:=\;\Bigl\{\delta\in\Gamma:\ \exists f,f'\in F \text{ with } \{f,\delta f'\}\in E\Bigr\}.
\]
The set $D$ is finite because $F$ is finite and $G$ is locally finite (only finitely many edges
emanate from vertices in $F$). Let
\[
M \;:=\; \max_{\delta\in D}\|\delta\|_\infty \;<\;\infty,
\]
where $\|\cdot\|_\infty$ denotes the $\ell_\infty$-norm under the identification $\Gamma\cong\mathbb{Z}^d$.

\smallskip
\textbf{Step 4: Points deep inside the box are far from the complement.}
Fix $r\in\mathbb{N}$ and set $s:=rM$. If $\gamma\in B_n$ satisfies
$\mathrm{dist}_\infty(\gamma,\Gamma\setminus B_n)>s$, then for any $f\in F$ we claim that
\[
\mathrm{dist}_G(\gamma f,\; V\setminus \Lambda_n) \;>\; r.
\]
Indeed, moving along one graph edge changes the cell-coordinate by an element of $D$, hence by
at most $M$ in $\ell_\infty$-norm. Therefore any path of length $\le r$ starting at $\gamma f$
remains within cell-coordinates contained in $B_n$, so it stays inside $\Lambda_n$.

Consequently, the $r$-thick boundary is contained in the union of cells whose coordinates lie
within $\ell_\infty$-distance $\le s$ of the box boundary:
\[
\partial^{(r)}\Lambda_n \;\subseteq\; \bigcup_{\gamma\in \mathrm{bd}_n^{(s)}} \gamma F,
\qquad
\mathrm{bd}_n^{(s)}:=\{\gamma\in B_n:\ \mathrm{dist}_\infty(\gamma,\Gamma\setminus B_n)\le s\}.
\]
Hence
\[
|\partial^{(r)}\Lambda_n|
\;\le\;
|F|\,|\mathrm{bd}_n^{(s)}|.
\]

\smallskip
\textbf{Step 5: Boundary-to-volume ratio tends to zero.}
Since $\mathrm{bd}_n^{(s)} = B_n\setminus B_{n-s}$ for $n>s$, we have
\[
|\mathrm{bd}_n^{(s)}|
\;=\;
|B_n|-|B_{n-s}|
\;=\;
(2n+1)^d - (2(n-s)+1)^d
\;=\; O(n^{d-1}).
\]
Therefore,
\[
\frac{|\partial^{(r)}\Lambda_n|}{|\Lambda_n|}
\;\le\;
\frac{|F|\,O(n^{d-1})}{|F|(2n+1)^d}
\;=\; O\!\left(\frac{1}{n}\right)\xrightarrow[n\to\infty]{}0.
\]
Since $r$ was arbitrary, $(\Lambda_n)$ is a Følner--van Hove sequence.
\end{proof}

\subsection{Monotonicity of the fixed-depth connective constant}
\label{subsec:monotonicity_cc}

We will prove the lemma~\ref{lem:muinf_le_sigma} in this subsection.

\begin{lemma}[SAW monotonicity under taking sub-graphs]
\label{lem:saw_monotonicity}
Let $G$ be a graph and let $H\subseteq G$ be a (finite) induced sub-graph. Fix a vertex $v\in V(H)$ and an integer
$k\ge 0$. Then the number of SAWs of length $k$ starting at $v$ and staying inside $H$ is at most the number
of SAWs of length $k$ starting at $v$ in $G$.
\end{lemma}

\begin{proof}
Every SAW in $H$ is also a SAW in $G$ (since $H$ is induced). Thus the set of SAWs in $H$ is a subset of the set
of SAWs in $G$, so the count cannot increase.
\end{proof}


\begin{proof}[proof of lemma \ref{lem:muinf_le_sigma}]
    By Lemma~\ref{lem:saw_monotonicity}, $\forall k\ge1$, for any $H\in\mathcal{H}_G$ we have 
    $$\#\{\text{SAW of length $k$ in $H$}\}\le\#\{\text{SAW of length $k$ in $G$}\}$$
    Then we have 
    $$\mu_{\inf}(\mathcal{H}_G)=\inf_{k\ge1}\mu_{k}(\mathcal{H}_G)=\inf_{k\ge1}\sup_{H\in\mathcal{H}_G}\mu_{k}(H)\le\inf_{k\ge1}\mu_{k}(G)=\mu_{\inf}(G)$$
    From the definition of $\limsup$ we have $\forall\varepsilon>0$, $\exists K>0$ s.t. $\forall k\ge K$,
    $$\mathcal{N}_{\leq k}(G,v)\le\mathcal{N}_{\leq K-1}(G,v)+\sum^k_{i=K}(\sigma(G)+\varepsilon)^i\le M(\sigma(G)+\varepsilon)^k$$
    where $M<\infty$ is a constant, let $\varepsilon\downarrow0$, it turns out $\mu_{\inf}(G)\le\mu_k(G)\le\sigma(G)$, and $\mu_{\inf}(\mathcal{H}_G)\leq \sigma(G)$.

    As for infinite lattice, by the property of quasi-transitive from \cite{grimmett2015bounds}, we have $\sigma(G)=\lim_{k\to\infty}\sigma_k(G)=\inf_{k\ge1}\sigma_k(G)\le\inf_{k\ge1}\mu_k(G)=\inf_{k\ge1}\mu_k(\mathcal{H}_G)$. By the squeeze theorem, $\sigma(G)=\mu_{\inf}(\mathcal{H}_G)$ holds.
\end{proof}

\begin{definition}[Quasi-transitive graph]
Let $G=(V,E)$ be a graph and let $\mathrm{Aut}(G)$ denote its automorphism group acting on $V$.
We say that $G$ is \emph{quasi-transitive} if $\mathrm{Aut}(G)$ has finitely many orbits on $V$.

Equivalently, $G$ is quasi-transitive if there exists a finite set $F\subset V$ such that
\[
V \;=\; \mathrm{Aut}(G)F \;:=\; \{\delta f: \delta\in\mathrm{Aut}(G),\ f\in F\}.
\]
\end{definition}

\begin{lemma}[Assumption $\mathcal{A}$ implies quasi-transitivity]
\label{lem:assumptionA_implies_quasitransitive}
Let $G=(V,E)$ be a connected, locally finite infinite graph satisfying Assumption $\mathcal{A}$, i.e.\ there exist an integer $d\ge 1$ and a subgroup $\Gamma\le \mathrm{Aut}(G)$ such that $\Gamma\cong \mathbb{Z}^d$ acts freely on $V$
and there exists a finite set $F\subset V$ with $V=\Gamma F$.
Then $G$ is quasi-transitive, i.e.\ $\mathrm{Aut}(G)$ has finitely many orbits on $V$.
\end{lemma}

\begin{proof}
Since $V=\Gamma F$ with $F$ finite, the $\Gamma$-action has at most $|F|$ vertex-orbits:
indeed,
\[
V \;=\; \bigcup_{f\in F} \Gamma f,
\]
so every vertex lies in the $\Gamma$-orbit of some $f\in F$.
Because $\Gamma\le \mathrm{Aut}(G)$, for every $v\in V$ we have $\Gamma v \subseteq \mathrm{Aut}(G)\,v$.
Hence each $\mathrm{Aut}(G)$-orbit is a union of $\Gamma$-orbits, and therefore the number of $\mathrm{Aut}(G)$-orbits is at most the number of $\Gamma$-orbits.
Consequently $\mathrm{Aut}(G)$ has at most $|F|$ orbits on $V$, and $G$ is quasi-transitive.
\end{proof}

\subsection{Uniform zero-freeness for all finite sub-graph}

\begin{lemma}[Uniform zero-freeness for all finite sub-graphs]
\label{lem:uniform_zerofree_all_finite_subgraphs}
Let $G=(V,E)$ be an infinite graph and fix $k\in\mathbb{N}$, statement holds:
for every $d>0$ and every $\lambda\in[0,\lambda_c(\mu_k(G)))$ there exists $\delta=\delta(d,k,\lambda)>0$
such that for every finite graph $H\subset G$,
the hard-core partition function $Z_H(\cdot)$ has no zeros in the complex neighborhood $\mathcal{L}_{\delta}([0,\lambda])$
\end{lemma}

\begin{proof}
    From the proof of lemma \ref{lem:muinf_le_sigma} we know that for every finite sub-graph $H\subset G$ one has
    $$\mu_k(H)\le \mu_k(G),$$
    and that the critical function $\lambda_c(\cdot)$ is nonincreasing, which means $\lambda_c(\mu_k(G))\le\lambda_c(\mu_k(H))$
    Let $\lambda\in[0,\lambda_c(\mu_k(G)))$, since we already have a zero-free regime at any $\lambda\in[0,\lambda_c(\mu_k(H)))\supset[0,\lambda_c(\mu_k(G)))$ for all finite sub-graph $H$ from theorem \ref{thm:zero-free}. Let $\mathcal{U}_H$ denote the zero-free regime of finite sub-graph $H$ contains the interval $[0,\lambda]$, take
    $$\mathcal{L}_{\delta}=\bigcap_{\substack{\text{finite H}\\H\subset G}}\mathcal{U}_H$$
    then $Z_H(\cdot)\neq0$ holds for all finite sub-graph $H$ at $\mathcal{L}_{\delta}$.
\end{proof}

\subsection{Analyticity of the infinite-volume free energy on a zero-free domain}

In this subsection we study the infinite-volume free energy and its analyticity. 
Following \cite{scott2005repulsive}, we define the free energy in the thermodynamic limit via an exhaustion by a Følner--van Hove sequence. 
Using the existence of such sequences under Assumption $\mathcal{A}$, we then establish that the resulting limit exists (and is independent of the chosen F\o lner-van Hove sequence), providing a well-defined notion of infinite-volume free energy for our model.

\begin{theorem}[Thermodynamic limit~\cite{scott2005repulsive}]\label{thm:thermodynamic_limit_periodic_graph}
Let $G=(V,E)$ be a connected, locally finite, infinite graph.
If $G$ holds for \textit{Assumption $\mathcal{A}$}
, consider the hard-core model on $G$ with constant fugacity $\lambda\ge 0$: for every finite
$\Lambda\subset V$ let
\[
Z_\Lambda(\lambda)\;=\;\sum_{\substack{I\subseteq \Lambda\\ I\ \mathrm{independent\ in}\ G[\Lambda]}}
\lambda^{|I|},
\qquad
f_\Lambda(\lambda)\;=\;\frac{1}{|\Lambda|}\log Z_\Lambda(\lambda).
\]
Then the limit
\[
f_\infty(\lambda)\;:=\;\lim_{n\to\infty} f_{\Lambda_n}(\lambda)
\]
exists along every F\o lner--van Hove sequence $(\Lambda_n)$ and is independent of the chosen sequence.
Moreover,
\[
0\le f_\infty(\lambda)\le \log(1+\lambda),
\]
and $f_\infty(\lambda)$ is increasing in $\lambda$ and convex as a function of $\log\lambda$ on $(0,\infty)$.
\end{theorem}

\begin{proof}
Fix once and for all a finite set $F\subset V$ of $\Gamma$-orbit representatives so that
$V=\bigcup_{\gamma\in\Gamma}\gamma F$.
Fix an isomorphism $\Gamma\cong \mathbb{Z}^d$. Let $a,c\in\mathbb{N}$ and set $L:=a+c$. Define the boxes
\[
Q_L:=\{0,1,\dots,L-1\}^d,\qquad Q_a:=\{0,1,\dots,a-1\}^d\subset Q_L,
\]
viewed as subsets of $\Gamma\cong\mathbb{Z}^d$. Let $L\Gamma:=L\mathbb{Z}^d\subset \Gamma$.
Then $\Gamma$ is partitioned as
\[
\Gamma=\bigcup_{t\in L\Gamma}(t+Q_L).
\]
For each $t\in L\Gamma$ define the corresponding vertex blocks and their cores
\[
C_t:=\bigcup_{\gamma\in t+Q_L}\gamma F,\qquad
C_t':=\bigcup_{\gamma\in t+Q_a}\gamma F\subset C_t.
\]
Then $V=\bigcup_{t\in L\Gamma}C_t$. Define
\[
D:=\Bigl\{\delta\in\Gamma:\ \exists f,f'\in F \text{ with } \{f,\delta f'\}\in E\Bigr\},
\qquad
M:=\max_{\delta\in D}\|\delta\|_\infty<\infty.
\]
(As in Lemma~\ref{lem:P_implies_vanhove}, $D$ is finite because $F$ is finite and $G$ is locally finite.)
Choose $c\ge 2M$. We claim that distinct cores have no edges between them:
for $t\neq t'$,
\begin{equation}\label{eq:no_edges_between_cores}
E(C_t',C_{t'}')=\varnothing.
\end{equation}
Indeed, take $u\in C_t'$ and $v\in C_{t'}'$. Write $u=\gamma f$ and $v=\gamma'f'$ with
$\gamma\in t+Q_a$ and $\gamma'\in t'+Q_a$. Since $t$ and $t'$ differ by a vector all of whose
coordinates are multiples of $L=a+c$, we have $\|\gamma-\gamma'\|_\infty\ge c$.
If $\{u,v\}\in E$, then by $\Gamma$-invariance of $E$ we would have
$\{f,\gamma^{-1}\gamma' f'\}\in E$, so $\gamma^{-1}\gamma'\in D$ and hence
$\|\gamma-\gamma'\|_\infty=\|\gamma^{-1}\gamma'\|_\infty\le M$, a contradiction because $c\ge 2M$.
This proves~\eqref{eq:no_edges_between_cores}.

Let $\Lambda\subset V$ be finite. Define the set of fully contained blocks
\[
\mathcal{T}_\Lambda:=\{t\in L\Gamma:\ C_t\subseteq \Lambda\},
\]
and the remaining vertices not covered by the selected cores,
\[
\Lambda_0:=\Lambda\setminus \bigcup_{t\in\mathcal{T}_\Lambda} C_t'.
\]
Because of~\eqref{eq:no_edges_between_cores}, any choice of independent sets inside distinct
cores can be combined into an independent set in $\Lambda$ (by taking the union), and there are
no constraints between different cores. Therefore we have the lower bound
\begin{equation}\label{eq:lower_bound_partition}
Z_\Lambda(\lambda)\ \ge\ \prod_{t\in\mathcal{T}_\Lambda} Z_{C_t'}(\lambda),
\end{equation}
obtained by forcing $\Lambda_0$ to be empty.

For the upper bound, note that for any finite set $S\subset V$ we trivially have
\begin{equation}\label{eq:trivial_upper_bound}
Z_S(\lambda)\ \le\ (1+\lambda)^{|S|},
\end{equation}
since each vertex is either occupied (weight $\lambda$) or vacant (weight $1$), ignoring exclusion
constraints. Hence
\begin{equation}\label{eq:upper_bound_partition}
Z_\Lambda(\lambda)
\ \le\
\Bigl(\prod_{t\in\mathcal{T}_\Lambda} Z_{C_t'}(\lambda)\Bigr)\,(1+\lambda)^{|\Lambda_0|}.
\end{equation}
Taking logarithms, \eqref{eq:lower_bound_partition}--\eqref{eq:upper_bound_partition} yield
\begin{equation}\label{eq:almost_additivity}
\Bigl|\log Z_\Lambda(\lambda)-\sum_{t\in\mathcal{T}_\Lambda}\log Z_{C_t'}(\lambda)\Bigr|
\ \le\ \log(1+\lambda)\,|\Lambda_0|.
\end{equation}

Since $C_t'$ is the image of $C_0'$ under the automorphism $t\in\Gamma$ and the fugacity is constant,
\[
Z_{C_t'}(\lambda)=Z_{C_0'}(\lambda)\qquad \forall\,t\in L\Gamma.
\]
Hence the sum in~\eqref{eq:almost_additivity} equals
\[
\sum_{t\in\mathcal{T}_\Lambda}\log Z_{C_t'}(\lambda)
=|\mathcal{T}_\Lambda|\cdot \log Z_{C_0'}(\lambda).
\]

Let $\mathrm{diam}(C_0)$ denote the graph diameter of the finite set $C_0$:
\[
\mathrm{diam}(C_0):=\max\{dist_G(x,y):x,y\in C_0\}<\infty.
\]

Define 

\begin{equation}
    \Lambda^{core}_0:=\Lambda\cap\bigcup_{t\in\mathcal{T}_\Lambda}(C_t\backslash C'_t), \qquad \Lambda^{bd}_0:=\Lambda\cap\bigcup_{t\notin\mathcal{T}_\Lambda}C_t.
\end{equation}

it's obvious $\Lambda_0\subset\Lambda_0^{core}\cup\Lambda_0^{bd}$ and $\Lambda_0^{core}\cap\Lambda_0^{bd}=\emptyset$, thus

\begin{equation}
    |\Lambda_0|\le|\Lambda_0^{core}|+|\Lambda_0^{bd}|
\end{equation}

For every $t\in\mathcal{T}_\Lambda$, since $C_t\subset\Lambda$,
\begin{equation}
    \Lambda\cap(C_t\backslash C'_t)=(C_t\backslash C'_t)
\end{equation}
and
\begin{equation}
    |\Lambda_0^{core}|=\sum_{t\in\mathcal{T}_\Lambda}|C_t\backslash C'_t|=|\mathcal{T}_\Lambda|\cdot |C_t\backslash C'_t|
\end{equation}

It turns out $\frac{|\Lambda_0^{core}|}{|\Lambda|}\le|\mathcal{T}_\Lambda|\cdot\bigl(1-\bigl(\frac{a}{a+c}\bigr)^d\bigr)$.

For $\Lambda^{bd}_0$, which combined by the block $C_t\not\subset\Lambda$, there is at least one vertex that

\begin{equation}
    \forall t\notin\mathcal{T}_\Lambda, \exists u_t\in C_t\cap(V\backslash\Lambda).
\end{equation}

for any $v\in\Lambda^{bd}_0$ that in the same $C_t$ with $u_t$, and

\begin{equation}
    dist_G(v, u_t)\le\mathrm{diam}(C_t)=\mathrm{diam}(C_0)=:R<\infty
\end{equation}

which means $v\in\partial^{(R)}\Lambda$ and $\Lambda^{bd}_0\subset\partial^{(R)}\Lambda\Rightarrow |\Lambda^{bd}_0|\le|\partial^{(R)}\Lambda|$.

After all, for any finite $\Lambda\subset V$:
\begin{align}
    &|\Lambda_0|\le|\mathcal{T}_\Lambda|\cdot |C_t\backslash C'_t|+|\partial^{(R)}\Lambda|\\
    \Rightarrow&\frac{|\Lambda_0|}{|\Lambda|}\le\frac{|\mathcal{T}_\Lambda|\cdot |C_t\backslash C'_t|}{|\Lambda|}+\frac{|\partial^{(R)}\Lambda|}{|\Lambda|}\\
    \Rightarrow&\frac{|\Lambda_0|}{|\Lambda|}\le|\mathcal{T}_\Lambda|\cdot\bigl(1-\bigl(\frac{a}{a+c}\bigr)^d\bigr)+\frac{|\partial^{(R)}\Lambda|}{|\Lambda|}
\end{align}

Let $(\Lambda_n)$ be any Følner-van Hove sequence. Apply~\eqref{eq:almost_additivity} with $\Lambda=\Lambda_n$:
\begin{equation}\label{eq:ratio_estimate}
\left|
f_{\Lambda_n}(\lambda)
-
\frac{|\mathcal{T}_{\Lambda_n}|}{|\Lambda_n|}\cdot \log Z_{C_0'}(\lambda)
\right|
\ \le\
\log(1+\lambda)\,\frac{|\Lambda_{n,0}|}{|\Lambda_n|}.
\end{equation}
We have $|\Lambda_{n,0}|\le |\partial^{(R)}\Lambda_n|$, hence the
right-hand side of~\eqref{eq:ratio_estimate} tends to $0$ by the van Hove property.

It remains to understand $|\mathcal{T}_{\Lambda_n}|/|\Lambda_n|$. Since $\{C_t\}_{t\in L\Gamma}$ is a partition
of $V$ into blocks of equal finite size $|C_0|$, we have
\[
|C_0|\cdot |\mathcal{T}_{\Lambda_n}|
\ \le\
|\Lambda_n|
\ \le\
|C_0|\cdot |\mathcal{T}_{\Lambda_n}| + |\Lambda_{n,0}|.
\]
Dividing by $|\Lambda_n|$ using $|\Lambda_{n,0}|/|\Lambda_n|\xrightarrow[a\to\infty]{\text{fix $c$}} 0$, we obtain
\[
\frac{|C_0|\cdot |\mathcal{T}_{\Lambda_n}|}{|\Lambda_n|}\longrightarrow 1,
\qquad\text{hence}\qquad
\frac{|\mathcal{T}_{\Lambda_n}|}{|\Lambda_n|}\longrightarrow \frac{1}{|C_0|}.
\]
Returning to~\eqref{eq:ratio_estimate}, we conclude that $f_{\Lambda_n}(\lambda)$ converges to
\[
f_\infty(\lambda)\ :=\ \frac{1}{|C_0|}\log Z_{C_0'}(\lambda).
\]
In particular, the limit exists.

The right-hand side expression above depends only on $a$ and $c$ through the fixed core $C_0'$ and block $C_0$,
and not on the particular Følner-van Hove sequence. Hence any two Følner-van Hove sequences yield the same limit,
so the limit is independent of the choice of sequence.
\end{proof}

\begin{proposition}[Vitali convergence~\cite{scott2005repulsive}]
\label{prop:vitali-scott-sokal}
Let $D\subset\mathbb{C}$ be a domain and let $\{f_n\}_{n\ge 1}$ be a sequence of analytic functions on $D$.
Assume that:
\begin{enumerate}
  \item[(a)] There exists a constant $M\in\mathbb{R}$ such that
  \[
    \Re f_n(z) \le M \qquad \text{for all } n\ge 1 \text{ and all } z\in D .
  \]
  \item[(b)] There exists a subset $S\subset D$ having at least one accumulation point in $D$ such that the pointwise limit
  \[
    f_\infty(z):=\lim_{n\to\infty} f_n(z)
  \]
  exists (and is finite) for every $z\in S$.
\end{enumerate}
Then $f_n$ converges uniformly on compact subsets of $D$ to an analytic function $f$ on $D$.
Moreover, $f$ is the unique analytic function on $D$ whose restriction to $S$ agrees with the pointwise limit $f_\infty$.
\end{proposition}

\begin{proof}[proof of theorem \ref{thm:infinite_threshold_strong}]
    now we prove the theorem \ref{thm:infinite_threshold_strong}, For any $k\in\mathbb{N}$, from lemma \ref{lem:uniform_zerofree_all_finite_subgraphs} we know that for every $\lambda\in[0, \lambda_c(\mu_k(G))$, there exists a uniform zero-free regime $\mathcal{L}_{\delta}$ around the interval $[0, \lambda]$ that $Z_H(\cdot)\neq 0$ holds for all finite sub-graph $H$.

    From theorem \ref{thm:thermodynamic_limit_periodic_graph} we know that the normalized free energy limit
    \begin{equation}
        f_\infty(\lambda)\;:=\;\lim_{n\to\infty} f_{\Lambda_n}(\lambda)
    \end{equation}
    exists along every F\o lner-van Hove sequence $(\Lambda_n)$ and is independent of the chosen sequence for every $\lambda$ at the positive real axis.

    There is $|Z_{\Lambda_n}(z)|\le(1+|z|)^{|\Lambda_n|}$, hence for any compact $K\Subset\mathcal{L}_{\delta}$, $\Re f\le\sup_K|z|$. We take $S=[0,\lambda]$ where $f_{\Lambda_n}$ pointwise converge to $f_\infty$. Then from proposition \ref{prop:vitali-scott-sokal} we know that $f_n$ converges uniformly to a unique analytic function $f_G$ on any compact subsets of $\mathcal{L}_{\delta}$.

    From lemma \ref{lem:assumptionA_implies_quasitransitive} we know that infinite graph $G$ satisfying assumption $\mathcal{A}$ is a quasi-transitive graph. As lemma \ref{lem:muinf_le_sigma} says, $G$ has $\mu_{\inf}(\mathcal{H}_G)\le\sigma(G)$, we can always find a $k\in\mathbb{N}$ s.t. $\sigma(G)\ge\mu_k(\mathcal{H}_G)$ and $\lambda_c(\sigma(G))\le\lambda_c(\mu_k(\mathcal{H}_G))$, which complete the proof.
    
\end{proof}



\subsection{Improved threshold on $\mathbb{Z}^2$ via Weitz trees (Sinclair et al.~\cite{sinclair2017spatial}, Appendix A)}
\label{subsec:2.538}
We now explain how to incorporate the improvement from Sinclair et al.~\cite{sinclair2017spatial}\ (Appendix~A) into our framework.

Let $T_{\mathrm{Weitz}}(\mathbb{Z}^2)$ be the Weitz SAW tree associated to $\mathbb{Z}^2$ under the ordering scheme
of \cite[App.~A]{sinclair2017spatial}. Then the connective constant (defined as a $\limsup$ growth rate) of this tree is at
most $2.429$, which yields the improved activity threshold $2.538$ for the hard-core model on $\mathbb{Z}^2$.

Let $H\subseteq \mathbb{Z}^2$ be any finite induced sub-graph, and let $T_{\mathrm{Weitz}}(H)$ be its Weitz SAW tree
constructed using the same ordering scheme as for $\mathbb{Z}^2$.
Then for every $k\ge 1$,
\begin{equation}
  \mu_k\!\bigl(T_{\mathrm{Weitz}}(H)\bigr)\;\le\; \mu_k\!\bigl(T_{\mathrm{Weitz}}(\mathbb{Z}^2)\bigr).
\end{equation}

Any $k$-length SAW in $T_{\mathrm{Weitz}}(H)$ corresponds to a $k$-length self-avoiding exploration in $H$ under
the same pruning/ordering rules. Since $H$ is a sub-graph of $\mathbb{Z}^2$, every such exploration is also feasible
in $\mathbb{Z}^2$ under the same rules, hence the number of SAWs of each length in $T_{\mathrm{Weitz}}(H)$ is at most
that in $T_{\mathrm{Weitz}}(\mathbb{Z}^2)$. The inequality for $\mu_k$ follows by the same monotonicity reasoning as
Lemma~\ref{lem:muinf_le_sigma}.

\begin{theorem}[Analytic free energy on $\mathbb{Z}^2$ up to the threshold $2.538$]
\label{thm:z2_2538}
For every $\lambda<2.538$, there exists a complex neighborhood $D$ of $[0,\lambda]$ such that $Z_H(z)\neq 0$ for all
finite induced sub-graphs $H\subseteq \mathbb{Z}^2$ and all $z\in D$. Consequently, along any exhaustion $(H_n)$ of
$\mathbb{Z}^2$ for which the real free energy density converges on a real interval in $D$, the infinite-volume
normalized free energy $f_{\mathbb{Z}^2}(z)$ exists and is analytic and unique on $D$.
\end{theorem}

\begin{proof}
As the (tree) connective constant of $T_{\mathrm{Weitz}}(\mathbb{Z}^2)$ is bounded by $2.429$ in the $\limsup$ sense. Hence ,
\begin{equation}
  \mu_{\inf}\!\bigl(T_{\mathrm{Weitz}}(\mathbb{Z}^2)\bigr)\;\le\;2.429,
\end{equation}
Then it gives the same bound uniformly for all finite
$H\subseteq\mathbb{Z}^2$:
\begin{equation}
  \mu_{\inf}\!\bigl(T_{\mathrm{Weitz}}(H)\bigr)\;\le\;2.429 \qquad \text{for all } H\subseteq\mathbb{Z}^2.
\end{equation}
Applying our finite-graph zero-free theorem to these Weitz-tree bounds yields a single domain $D$ (depending on
$\varepsilon$) that is zero-free uniformly for all finite induced sub-graphs $H\subseteq\mathbb{Z}^2$ at activities
$\lambda<\lambda_c(2.429)$. It yields the regime up to
$\lambda<\lambda_c(2.429)=2.538$, as in \cite[App.~A]{sinclair2017spatial}. 

Finally, since $D$ is uniform over all finite sub-graphs, analyticity of the infinite-volume normalized free energy
on $D$ follows our above results. 
\end{proof}

Moreover, the value $2.429$ is a numerical upper bound and is not known to be tight; hence the threshold could be improved by sharpening the underlying computation. The same principle applies to $\mathbb{Z}^d$: better algorithmic upper bounds would yield better thresholds on the infinite lattice.

\bibliographystyle{alpha}
\bibliography{main}

@article{ChudnovskySeymour07,
  title = {The roots of the independence polynomial of a clawfree graph},
  author = {Chudnovsky, Maria and Seymour, Paul},
  journal = {Journal of Combinatorial Theory, Series B},
  volume = {97},
  number = {3},
  pages = {350--357},
  year = {2007},
  publisher = {Elsevier},
  timestamp = {Fri, 07 Jun 2024 01:00:00 +0200},
  biburl = {https://dblp.org/rec/journals/jct/ChudnovskyS07.bib},
  bibsource = {dblp computer science bibliography, https://dblp.org},
  doi = {10.1016/j.jctb.2006.06.001},
  _bib2doi_selected = {dblp:/rec/journals/jct/ChudnovskyS07.bib},
  _bib2doi_confirmed = {true},
}

@article{jerrum1993polynomial,
  title = {{Polynomial-time approximation algorithms for the Ising model}},
  author = {Jerrum, Mark and Sinclair, Alistair},
  journal = {SIAM Journal on computing},
  volume = {22},
  number = {5},
  pages = {1087--1116},
  year = {1993},
  publisher = {SIAM},
  timestamp = {Sun, 02 Jun 2019 01:00:00 +0200},
  biburl = {https://dblp.org/rec/journals/siamcomp/JerrumS93.bib},
  bibsource = {dblp computer science bibliography, https://dblp.org},
  doi = {10.1137/0222066},
  _bib2doi_selected = {dblp:/rec/journals/siamcomp/JerrumS93.bib},
  _bib2doi_confirmed = {true},
}

@article{sinclair2014approximation,
  title = {{Approximation algorithms for two-state anti-ferromagnetic spin systems on bounded degree graphs}},
  author = {Sinclair, Alistair and Srivastava, Piyush and Thurley, Marc},
  journal = {Journal of Statistical Physics},
  volume = {155},
  number = {4},
  pages = {666--686},
  year = {2014},
  publisher = {Springer},
}

@article{heilmann1972theory,
  title = {Theory of monomer-dimer systems},
  author = {Heilmann, Ole J and Lieb, Elliott H},
  journal = {Communications in mathematical Physics},
  volume = {25},
  number = {3},
  pages = {190--232},
  year = {1972},
  publisher = {Springer},
}

@book{barvinok2016combinatorics,
  title = {Combinatorics and complexity of partition functions},
  author = {Barvinok, Alexander},
  volume = {30},
  year = {2016},
  publisher = {Springer},
  timestamp = {Tue, 26 Sep 2017 01:00:00 +0200},
  biburl = {https://dblp.org/rec/books/daglib/0041359.bib},
  bibsource = {dblp computer science bibliography, https://dblp.org},
  doi = {10.1007/978-3-319-51829-9},
  _bib2doi_selected = {dblp:/rec/books/daglib/0041359.bib},
  _bib2doi_confirmed = {true},
}

@article{patel2017deterministic,
  title = {Deterministic polynomial-time approximation algorithms for partition functions and graph polynomials},
  author = {Patel, Viresh and Regts, Guus},
  journal = {SIAM Journal on Computing},
  volume = {46},
  number = {6},
  pages = {1893--1919},
  year = {2017},
  publisher = {SIAM},
  timestamp = {Tue, 13 Mar 2018 00:00:00 +0100},
  biburl = {https://dblp.org/rec/journals/siamcomp/PatelR17.bib},
  bibsource = {dblp computer science bibliography, https://dblp.org},
  doi = {10.1137/16M1101003},
  _bib2doi_selected = {dblp:/rec/journals/siamcomp/PatelR17.bib},
  _bib2doi_confirmed = {true},
}

@article{peters2019conjecture,
  title = {On a conjecture of Sokal concerning roots of the independence polynomial},
  author = {Peters, Han and Regts, Guus},
  journal = {Michigan Mathematical Journal},
  volume = {68},
  number = {1},
  pages = {33--55},
  year = {2019},
  publisher = {University of Michigan, Department of Mathematics},
}

@article{liu2019fisher,
  title = {Fisher zeros and correlation decay in the Ising model},
  author = {Liu, Jingcheng and Sinclair, Alistair and Srivastava, Piyush},
  journal = {Journal of Mathematical Physics},
  volume = {60},
  number = {10},
  year = {2019},
  publisher = {AIP Publishing},
}

@article{peters2020location,
  title = {Location of zeros for the partition function of the Ising model on bounded degree graphs},
  author = {Peters, Han and Regts, Guus},
  journal = {Journal of the London Mathematical Society},
  volume = {101},
  number = {2},
  pages = {765--785},
  year = {2020},
  publisher = {Wiley Online Library},
}

@article{liu2022correlation,
  title = {Correlation decay and partition function zeros: Algorithms and phase transitions},
  author = {Liu, Jingcheng and Sinclair, Alistair and Srivastava, Piyush},
  journal = {SIAM Journal on Computing},
  number = {0},
  pages = {FOCS19--200},
  year = {2022},
  publisher = {SIAM},
}

@article{gamarnik2023correlation,
  title = {Correlation decay and the absence of zeros property of partition functions},
  author = {Gamarnik, David},
  journal = {Random Structures \& Algorithms},
  volume = {62},
  number = {1},
  pages = {155--180},
  year = {2023},
  publisher = {Wiley Online Library},
  timestamp = {Thu, 05 Jan 2023 00:00:00 +0100},
  biburl = {https://dblp.org/rec/journals/rsa/Gamarnik23.bib},
  bibsource = {dblp computer science bibliography, https://dblp.org},
  doi = {10.1002/rsa.21083},
  _bib2doi_selected = {dblp:/rec/journals/rsa/Gamarnik23.bib},
  _bib2doi_confirmed = {true},
}

@article{regts2023absence,
  title = {Absence of zeros implies strong spatial mixing},
  author = {Regts, Guus},
  journal = {Probability Theory and Related Fields},
  volume = {186},
  number = {1},
  pages = {621--641},
  year = {2023},
  publisher = {Springer},
}

@article{shao2024zero,
  title = {From Zero-Freeness to Strong Spatial Mixing via a Christoffel-Darboux Type Identity},
  author = {Shao, Shuai and Ye, Xiaowei},
  journal = {arXiv preprint arXiv:2401.09317},
  year = {2024},
}

@inproceedings{weitz2006counting,
  title = {Counting independent sets up to the tree threshold},
  author = {Weitz, Dror},
  booktitle = {Proceedings of the thirty-eighth annual ACM symposium on Theory of computing},
  pages = {140--149},
  year = {2006},
  timestamp = {Tue, 06 Nov 2018 00:00:00 +0100},
  biburl = {https://dblp.org/rec/conf/stoc/Weitz06.bib},
  bibsource = {dblp computer science bibliography, https://dblp.org},
  doi = {10.1145/1132516.1132538},
  _bib2doi_selected = {dblp:/rec/conf/stoc/Weitz06.bib},
  _bib2doi_confirmed = {true},
}

@article{shao2021contraction,
  title = {Contraction: A unified perspective of correlation decay and zero-freeness of 2-spin systems},
  author = {Shao, Shuai and Sun, Yuxin},
  journal = {Journal of Statistical Physics},
  volume = {185},
  pages = {1--25},
  year = {2021},
  publisher = {Springer},
}

@article{bandyopadhyay2008counting,
  title = {Counting without sampling: Asymptotics of the log-partition function for certain statistical physics models},
  author = {Bandyopadhyay, Antar and Gamarnik, David},
  journal = {Random Structures \& Algorithms},
  volume = {33},
  number = {4},
  pages = {452--479},
  year = {2008},
  publisher = {Wiley Online Library},
  timestamp = {Fri, 26 May 2017 01:00:00 +0200},
  biburl = {https://dblp.org/rec/journals/rsa/BandyopadhyayG08.bib},
  bibsource = {dblp computer science bibliography, https://dblp.org},
  doi = {10.1002/rsa.20236},
  _bib2doi_selected = {dblp:/rec/journals/rsa/BandyopadhyayG08.bib},
  _bib2doi_confirmed = {true},
}

@article{shao2025zero,
  title = {Zero-Freeness is All You Need: A Weitz-Type FPTAS for the Entire Lee-Yang Zero-Free Region},
  author = {Shao, Shuai and Shi, Ke},
  journal = {arXiv preprint arXiv:2509.06623},
  year = {2025},
}

@article{sinclair2017spatial,
  title={Spatial mixing and the connective constant: Optimal bounds},
  author={Sinclair, Alistair and Srivastava, Piyush and {\v{S}}tefankovi{\v{c}}, Daniel and Yin, Yitong},
  journal={Probability Theory and Related Fields},
  volume={168},
  number={1},
  pages={153--197},
  year={2017},
  publisher={Springer}
}

@article{scott2005repulsive,
  title={The repulsive lattice gas, the independent-set polynomial, and the Lov{\'a}sz local lemma},
  author={Scott, Alexander D and Sokal, Alan D},
  journal={Journal of Statistical Physics},
  volume={118},
  number={5},
  pages={1151--1261},
  year={2005},
  publisher={Springer}
}

@inproceedings{hammersley1957percolation,
  title={Percolation processes: II. The connective constant},
  author={Hammersley, John M},
  booktitle={Mathematical Proceedings of the Cambridge Philosophical Society},
  volume={53},
  number={3},
  pages={642--645},
  year={1957},
  organization={Cambridge University Press}
}

@article{grimmett2015bounds,
  title={Bounds on connective constants of regular graphs},
  author={Grimmett, Geoffrey R and Li, Zhongyang},
  journal={Combinatorica},
  volume={35},
  number={3},
  pages={279--294},
  year={2015},
  publisher={Springer}
}

@article{duminil2012connective,
  title={The connective constant of the honeycomb lattice equals $\sqrt{2+\sqrt{2}}$},
  author={Duminil-Copin, Hugo and Smirnov, Stanislav},
  journal={Annals of Mathematics},
  pages={1653--1665},
  year={2012},
  publisher={JSTOR}
}

@article{lacoin2014non,
  title={Non-coincidence of quenched and annealed connective constants on the supercritical planar percolation cluster},
  author={Lacoin, Hubert},
  journal={Probability Theory and Related Fields},
  volume={159},
  number={3},
  pages={777--808},
  year={2014},
  publisher={Springer}
}

@article{runnels1965hard,
  title={Hard-square lattice gas},
  author={Runnels, LK},
  journal={Physical Review Letters},
  volume={15},
  number={14},
  pages={581},
  year={1965},
  publisher={APS}
}

@article{regts2018zero,
  title={Zero-free regions of partition functions with applications to algorithms and graph limits},
  author={Regts, Guus},
  journal={Combinatorica},
  volume={38},
  number={4},
  pages={987--1015},
  year={2018},
  publisher={Springer}
}

@article{shearer1985problem,
  title={On a problem of Spencer},
  author={Shearer, James B.},
  journal={Combinatorica},
  volume={5},
  number={3},
  pages={241--245},
  year={1985},
  publisher={Springer}
}

@inproceedings{bencs2023complex,
  title={On complex roots of the independence polynomial},
  author={Bencs, Ferenc and Csikv{\'a}ri, P{\'e}ter and Srivastava, Piyush and Vondr{\'a}k, Jan},
  booktitle={Proceedings of the 2023 annual ACM-SIAM symposium on discrete algorithms (SODA)},
  pages={675--699},
  year={2023},
  organization={SIAM}
}

@inproceedings{vera2013improved,
  title={Improved bounds on the phase transition for the hard-core model in 2-dimensions},
  author={Vera, Juan C and Vigoda, Eric and Yang, Linji},
  booktitle={International Workshop on Approximation Algorithms for Combinatorial Optimization},
  pages={699--713},
  year={2013},
  organization={Springer}
}

@article{blanca2019phase,
  title={Phase Coexistence for the Hard-Core Model on $\mathbb{Z}^2$},
  author={Blanca, Antonio and Chen, Yuxuan and Galvin, David and Randall, Dana and Tetali, Prasad},
  journal={Combinatorics, Probability and Computing},
  volume={28},
  number={1},
  pages={1--22},
  year={2019},
  publisher={Cambridge University Press}
}

@inproceedings{efthymiou2026sampling,
  title={On sampling two spin models using the local connective constant},
  author={Efthymiou, Charilaos},
  booktitle={Proceedings of the 2026 Annual ACM-SIAM Symposium on Discrete Algorithms (SODA)},
  pages={972--983},
  year={2026},
  organization={SIAM}
}

@article{lee1952statistical,
  title={Statistical theory of equations of state and phase transitions. II. Lattice gas and Ising model},
  author={Lee, Tsung-Dao and Yang, Chen-Ning},
  journal={Physical Review},
  volume={87},
  number={3},
  pages={410},
  year={1952},
  publisher={APS}
}

@article{jerrum2025zero,
  title={Zero-free regions for the independence polynomial on restricted graph classes},
  author={Jerrum, Mark and Patel, Viresh},
  journal={arXiv preprint arXiv:2510.01466},
  year={2025}
}

@article{jensen1998self,
  title={Self-avoiding walks, neighbour-avoiding walks and trails on semiregular lattices},
  author={Jensen, Iwan and Guttmann, Anthony J},
  journal={Journal of Physics A: Mathematical and General},
  volume={31},
  number={40},
  pages={8137},
  year={1998},
  publisher={IOP Publishing}
}

@article{ponitz2000improved,
  title   = {Improved Upper Bounds for Self-Avoiding Walks in {$\mathbb{Z}^d$}},
  author  = {P{\"o}nitz, Andr{\'e} and Tittmann, Peter},
  journal = {The Electronic Journal of Combinatorics},
  volume  = {7},
  pages   = {R21},
  year    = {2000}
}

@article{alm2005upper,
  title={Upper and lower bounds for the connective constants of self-avoiding walks on the Archimedean and Laves lattices},
  author={Alm, Sven Erick},
  journal={Journal of Physics A: Mathematical and General},
  volume={38},
  number={10},
  pages={2055},
  year={2005},
  publisher={IOP Publishing}
}

@article{luby1999fast,
	Author = {Luby, Michael and Vigoda, Eric},
	Journal = {Random Structures Algorithms},
	Number = {3-4},
	Pages = {229--241},
	Publisher = {Wiley Online Library},
	Title = {Fast convergence of the Glauber dynamics for sampling independent sets},
	Volume = {15},
	Year = {1999}
}

@article{hayes2006coupling,
	Author = {Hayes, Thomas P. and Vigoda, Eric},
	Fjournal = {The Annals of Applied Probability},
	Journal = {Ann. Appl. Probab.},
	Number = {3},
	Pages = {1297--1318},
	Title = {Coupling with the stationary distribution and improved sampling for colorings and independent sets},
	Volume = {16},
	Year = {2006},
}

@article{efthymiou2016convergence,
	Author = {Efthymiou, Charilaos and Hayes, Thomas P. and {\v{S}}tefankovi{\v{c}}, Daniel and Vigoda, Eric and Yin, Yitong},
	Fjournal = {SIAM Journal on Computing},
	Journal = {SIAM J. Comput.},
	Number = {2},
	Pages = {581--643},
	Title = {Convergence of {MCMC} and loopy {BP} in the tree uniqueness region for the hard-core model},
	Volume = {48},
	Year = {2019}
}

@inproceedings{anari2020spectral,
	Author = {Nima Anari and Kuikui Liu and Shayan {Oveis Gharan}},
	Booktitle = {FOCS},
	Pages = {1319--1330},
	Title = {Spectral Independence in High-Dimensional Expanders and Applications to the Hardcore Model},
	Year = {2020}
}

@inproceedings{chen2020optimal,
    AUTHOR = {Chen, Zongchen and Liu, Kuikui and Vigoda, Eric},
     TITLE = {Optimal mixing of {G}lauber dynamics: entropy factorization
              via high-dimensional expansion},
 BOOKTITLE = {STOC},
     PAGES = {1537--1550},
      YEAR = {2021},
   Journal = {arXiv preprint arXiv:2011.02075},
}

@inproceedings{ChenE22,
  author       = {Yuansi Chen and
                  Ronen Eldan},
  title        = {Localization Schemes: {A} Framework for Proving Mixing Bounds for
                  Markov Chains },
  booktitle    = {FOCS},
  pages        = {110--122},
  year         = {2022},
}

@inproceedings{Chen0YZ22,
  author       = {Xiaoyu Chen and
                  Weiming Feng and
                  Yitong Yin and
                  Xinyuan Zhang},
  title        = {Optimal mixing for two-state anti-ferromagnetic spin systems},
  booktitle    = {FOCS},
  pages        = {588--599},
  year         = {2022},
}

@article{chen2024rapid,
  title={Rapid mixing via coupling independence for spin systems with unbounded degree},
  author={Chen, Xiaoyu and Feng, Weiming},
  journal={arXiv preprint arXiv:2407.04672},
  year={2024}
}

@inproceedings{chen2023uniqueness,
  author       = {Xiaoyu Chen and
                  Jingcheng Liu and
                  Yitong Yin},
  title        = {Uniqueness and Rapid Mixing in the Bipartite Hardcore Model},
  booktitle    = {FOCS},
  pages        = {1991--2005},
  publisher    = {{IEEE}},
  year         = {2023},
  url          = {https://doi.org/10.1109/FOCS57990.2023.00121},
  doi          = {10.1109/FOCS57990.2023.00121},
  bibsource    = {dblp computer science bibliography, https://dblp.org}
}

@article{jerrum2024glauber,
  title={Glauber dynamics for the hard-core model on bounded-degree $ H $-free graphs},
  author={Jerrum, Mark},
  journal={arXiv preprint arXiv:2404.07615},
  year={2024}
}

@inproceedings{chen2024fast,
  title={Fast sampling of b-matchings and b-edge covers},
  author={Chen, Zongchen and Gu, Yuzhou},
  booktitle={Proceedings of the 2024 Annual ACM-SIAM Symposium on Discrete Algorithms (SODA)},
  pages={4972--4987},
  year={2024},
  organization={SIAM}
}

@article{restrepo2013improved,
  title={Improved mixing condition on the grid for counting and sampling independent sets},
  author={Restrepo, Ricardo and Shin, Jinwoo and Tetali, Prasad and Vigoda, Eric and Yang, Linji},
  journal={Probability Theory and Related Fields},
  volume={156},
  number={1},
  pages={75--99},
  year={2013},
  publisher={Springer}
}

@article{chen2025rapid,
  title={Rapid Mixing on Random Regular Graphs beyond Uniqueness},
  author={Chen, Xiaoyu and Chen, Zejia and Chen, Zongchen and Yin, Yitong and Zhang, Xinyuan},
  journal={arXiv preprint arXiv:2504.03406},
  year={2025}
}

@article{efthymiou2023mixing,
  title={On the Mixing Time of Glauber Dynamics for the Hard-Core and Related Models on G (n, d/n)},
  author={Efthymiou, Charilaos and Feng, Weiming},
  journal={arXiv preprint arXiv:2302.06172},
  year={2023}
}

@article{jacobsen2016growth,
  title={On the growth constant for square-lattice self-avoiding walks},
  author={Jacobsen, Jesper Lykke and Scullard, Christian R and Guttmann, Anthony J},
  journal={Journal of Physics A: Mathematical and Theoretical},
  volume={49},
  number={49},
  pages={494004},
  year={2016},
  publisher={IOP Publishing}
}

@article{slade1994self,
  title={Self-avoiding walks},
  author={Slade, Gordon},
  journal={The Mathematical Intelligencer},
  volume={16},
  number={1},
  pages={29--35},
  year={1994},
  publisher={Springer-Verlag New York}
}

\end{document}